\documentclass[12pt, draftclsnofoot, onecolumn, romanappendices]{IEEEtran}

\usepackage{cite}

\ifCLASSINFOpdf
   \usepackage[pdftex]{graphicx}
  \DeclareGraphicsExtensions{.pdf,.jpeg,.png}
\else
\fi

\usepackage[cmex10]{amsmath}
\usepackage{amssymb}
\usepackage{enumitem}
\usepackage{subcaption}
\usepackage[font=footnotesize,justification=centering]{caption}
\usepackage{epstopdf}
\usepackage{setspace}

\usepackage{amsthm}
\usepackage{url}

\newtheorem{theorem}{Theorem}
\newtheorem{corollary}{Corollary}
\newtheorem{lemma}{Lemma}
\newtheorem{proposition}{Proposition}
\newtheorem{definition}{Definition}
\newtheorem{remark}{Remark}
\newtheorem{example}{Example}

\DeclareMathOperator{\Mat}{Mat}
\DeclareMathOperator{\rk}{rank}
\DeclareMathOperator{\Nm}{Nm}
\DeclareMathOperator{\diag}{diag}
\DeclareMathOperator{\vect}{vec}
\DeclareMathOperator{\id}{id}
\DeclareMathOperator{\Aut}{Aut}
\DeclareMathOperator{\ord}{ord}
\DeclareMathOperator*{\argmin}{arg\,min}
\DeclareMathOperator{\SNR}{SNR}

\usepackage{tikz}
\usetikzlibrary{arrows, positioning}
\tikzset{
    >=stealth',
    user/.style={
           rectangle,
           rounded corners,
           draw,
           minimum height=2em,
           minimum width=1cm,
           text centered},
    relay/.style={
           rectangle,
           rounded corners,
           draw,
           minimum height=1em,
           minimum width=1cm,
           text centered},
    pil/.style={
           ->,
           shorten <=2pt,
           shorten >=2pt},
    pil_rev/.style={
           <-,
           shorten <=2pt,
           shorten >=2pt},
    pil_dash/.style={
    		->, dashed,
    		shorten <=2pt,
    		shorten >=2pt}
}

\begin{document}

\title{Fast-Decodable Space--Time Codes for the $N$-Relay and Multiple-Access MIMO Channel}

\author{Amaro~Barreal, Camilla~Hollanti,~\IEEEmembership{Member,~IEEE}, and~Nadya~Markin,~\IEEEmembership{Member,~IEEE,}
\thanks{A. Barreal and C. Hollanti are with the Department
of Mathematics and Systems Analysis, Aalto University, Finland (e-mail: firstname.lastname@aalto.fi). They are financially supported by the Academy of Finland grants \#276031, \#282938, and \#283262, and by a grant from Magnus Ehrnrooth Foundation, Finland. The support from the European Science Foundation under the COST Action IC1104 is also gratefully acknowledged.}
\thanks{N. Markin is with the School of Physical and Mathematical Sciences, Nanyang Technological University, Singapore (e-mail: nmarkin@ntu.edu). She is financially supported by the Singapore National Research Foundation under Research Grant NRF-RF2009-07.
Part of this research was carried out while N. Markin was visiting Aalto University, January 2012.}
\thanks{Parts of this paper were presented in MTNS12 \cite{hollanti_relay1} and ISIT12 \cite{hollanti_relay2}.}}

\maketitle

\begin{abstract}
	In this article, the first general constructions of fast-decodable, more specifically (conditionally) $g$-group decodable, space--time block codes for the Nonorthogonal Amplify and Forward (NAF) Multiple-Input Multiple-Output (MIMO) relay channel under the half-duplex constraint are proposed. In this scenario, the source and the intermediate relays used for data amplification are allowed to employ multiple antennas for data transmission and reception. The worst-case decoding complexity of the obtained codes is reduced by up to $75\%$. In addition to being fast-decodable, the proposed codes achieve full-diversity and have nonvanishing determinants, which has been shown to be useful for achieving the optimal Diversity-Multiplexing Tradeoff (DMT) of the NAF channel.
	
	Further, it is shown that the same techniques as in the cooperative scenario can be utilized to achieve fast-decodability for $K$-user MIMO Multiple-Access Channel (MAC) space--time block codes. The resulting codes in addition exhibit the conditional nonvanishing determinant property which, for its part, has been shown to be useful for achieving the optimal MAC-DMT.
\end{abstract}


\IEEEpeerreviewmaketitle

\section{Introduction}
\IEEEPARstart{T}{he} amount of data stored and the data traffic worldwide has reached incredible numbers. It was estimated that in 2011, $1800\cdot 10^{18}$ bytes of data needed to be stored worldwide, and astonishing $5200\cdot 10^{18}$ bytes of information have been created between January and November 1$^{\text{st}}$ 2014 \cite{emc}. The availability of such an astronomical amount of data and rapid progress in communications engineering and wireless communications explain the observed growth of mobile data traffic, which increased from $0.82\cdot 10^{18}$ bytes at the end of 2012 to $1.5\cdot 10^{18}$ bytes at the end of 2013, whereof $56\%$ of the traffic was mobile video traffic. In addition, about 526 million mobile devices and connections were added globally in 2013, and the number of mobile-connected devices will exceed the number of people on earth by the end of 2014 
\cite{cisco}.

These facts illustrate that networks will soon need to accommodate many new types of devices and be up to the enormous load while still living up to the expectations of exigent future users, which will be accessing any type of data from different devices at any time and from any corner of the world, demanding high reliability, reasonable speed, low energy consumption, etc. 

With this goal in mind, a tremendous effort is being made by both academic and industrial researchers focusing on the future 5$^{\text{th}}$ Generation (5G) wireless systems. Although many aspects still need to be discussed, as of today it is clear that 5G will consist of an integration of different techniques rather than being a single new technology, including distributed antenna systems and massive Multiple-Input Multiple-Output (MIMO) systems \cite{ericsson}.
Yet the various radio-access technologies to be included in 5G are only one side of the coin, as the overall performance of future networks will highly depend on the channel coding techniques employed. 

As a second motivating aspect, the issue of reliably storing the worldwide available data led to considering distributed storage systems, and a plethora of research has been done in the last few years regarding optimal storage codes, thus focusing on the network layer. However, an important aspect that needs consideration is data repair and reconstruction over wireless channels to provide flexibility and user mobility, even if the storage cloud itself would be wired, a feature related to the more general concept of \emph{wireless edge} \cite{rabaey1, rabaey2, rabaey3}. Many of the known algebraic physical layer communications techniques are however futile in this scenario due to the high decoding complexity they require. This calls for less complex coding techniques and transmission protocols, for instance introducing helping relays, as proposed in \cite{barreal_dss} or subsequent work \cite{hollanti_dss}.

\subsection{Related Work and Contributions}
\label{subsec:contributions}

Fast-Decodable (FD) codes are codes enjoying reduced complexity of Maximum-Likelihood (ML) decoding due to a smart inner structure allowing for parallelization in the ML search. 
First introduced in \cite{viterbo}, FD codes have been the subject of much interest \cite{biglieri, jithamithra, srinath, markin, vehkalahti, berhuy, berhuy2}. 
Fast-decodability of a code is achieved precisely when a subset of the generating matrices of the code satisfies certain {\emph{mutual orthogonality conditions}} \cite{ren, berhuy2}, which will be made explicit later on. 

Recent results \cite{berhuy2, berhuy, vehkalahti} show that fast-decodabilty imposes  constraints on the rate of the code on one hand, and on the algebraic parameters of the code, on the other. In particular, codes arising from division algebras (division resulting in full-diversity) can enjoy a reduction in decoding complexity order by a factor of no more than 4, \emph{i.e.}, a decoding complexity reduction of 75\%. In this paper, the best possible complexity reduction by a factor of 4 is in fact achieved.

On the other hand, the increasing interest in cooperative diversity techniques motivates the study of distributed codes. Since the introduction of the multiple-access relay channel \cite{kramer}, many protocols have been considered for data exchange in this scenario, such as the amplify-and-forward \cite{chen} or compute-and-forward \cite{nazer} protocol. Several distributed codes have been proposed \cite{laneman}, \cite{jing}, and it was in \cite{kiran} where the issue of the high decoding complexity of distributed codes was firstly addressed. Codes with low decoding complexity have then been constructed \emph{e.g.}, in \cite{rajan}, using Clifford algebras as the underlying structure. As far as the authors' are aware, all attempts to construct FD distributed codes, however, assume a single antenna at the source and the relays and furthermore do not achieve the Nonvanishing Determinant (NVD) property. A first framework for constructing NVD codes with reduced decoding complexity was proposed in \cite{hollanti_relay1}, and explicit examples were given in \cite{hollanti_relay1,hollanti_relay2}. In \cite{barreal_icmcta}, the authors of the present paper constructed the first FD distributed Space--Time (ST) codes for the MIMO channel and provided simulation results illustrating that imposing fast-decodability does not have an adverse effect on the performance of the codes. The constructions are however not general, but rather very specific example codes. 

The aim of this article is to provide methods to construct codes with desirable properties for good performance and reduced decoding complexity for flexible distributed and noncooperative multiuser physical layer MIMO communications.  
The main contributions of this article are:
\begin{itemize}
	\item Theorem~\ref{thm:single_antenna_relay_code}, which provides a method for constructing an infinite family of FD distributed ST codes having code rate 4 real symbols per channel use (rscu) and satisfying the NVD property for any number $N$ of relays. The theorem assumes a single antenna at the source and each of the relays. The resulting codes exhibit a worst-case decoding complexity $|S|^{5N}$ as opposed to $|S|^{8N}$ of a non-FD code of the same rank, where $S\subset \mathbb{Z}$ is the finite signaling alphabet used. The codes achieve the optimal Diversity-Multiplexing gain Tradeoff (DMT) for the relay channel when the destination has two antennas. 
		
	\item Theorem~\ref{thm:mult_antenna_relay_code}, resulting in a construction of an infinite family of FD distributed ST codes with code rate 2 rscu and the NVD property for $N = \frac{p-1}{2}$ relays, $p \ge 5$ prime. We assume  that the number of antennas $n_s$ at the source and the number of antennas $n_r$ at each relay satisfy $n_s+n_r = 4$, while one antenna suffices at the destination. The resulting codes have worst-case decoding complexity $|S|^{4N}$ or $|S|^{2N}$, respectively corresponding to a $50\%$ or $75\%$ reduction from the complexity $|S|^{8N}$ of a non-FD code of the same rank. According to a recent result \cite{berhuy}, $75\%$ is the best possible reduction for a division algebra based code. 
	
To the best of the authors' knowledge, we obtain the first distributed ST codes for multiple antennas which are FD and have the NVD property, excluding the example codes in \cite{barreal_icmcta}.
		
	\item Multiple explicit constructions, alongside simulation results disclosing the performance of FD distributed ST codes constructed using the methods derived in this article. 
	
	\item Extension of the results on FD distributed ST codes to the noncooperative MIMO Multiple-Access Channel (MAC) for an arbitrary number $K$ of users, resulting in FD codes for this scenario which achieve the conditional NVD property. 
	
	\item As a nontechnical contribution, the paper is written in a self-contained way, also providing a concise overview of algebraic FD ST-codes. 
\end{itemize}

With the exception Theorem~\ref{thm:single_antenna_relay_code}, which was presented as a preliminary result in \cite{hollanti_relay1}, this article contains exclusively novel theoretical and numerical results. 
 
The paper is organized as follows: We start with a recapitulation of ST codes in Section~\ref{sec:stc} and how they are constructed from cyclic division algebras. We also briefly study fast-decodability and the notions of conditional $g$-group and $g$-group decodability. Further, we recall the iterative ST code construction from cyclic algebras. In Section~\ref{sec:distributed} we propose two constructive methods to obtain FD ST codes with the NVD property for an $N$-relay channel under the half-duplex constraint, where the source and each of the relays are equipped with either one or multiple antennas. We then show in Section~\ref{sec:mac_stc} how these constructions can help obtain FD ST codes in the $K$-user MIMO-MAC. Section~\ref{sec:conclusion} concludes the paper.

\section{Space--Time Codes}
\label{sec:stc}

The increasing demand for user mobility observed during the past decades has motivated a plethora of research in the area of wireless communications. The change from the well-studied case of wired communications to data transmission over wireless networks called for novel coding techniques that were able to deal with the fading effects of wireless channels. From the start, algebraic and number theoretical tools have been proven useful for constructing well-performing codes, at first considering only devices at both ends of the channel equipped with a single antenna each. The rapid progress in communications engineering quickly led to considering multiple antennas at both ends of the channel for data rate increase. Considering this type of channels, known as \emph{multiple-input multiple-output} channels, \emph{space--time coding} was introduced as a promising technique for error prevention when transmitting information in the MIMO scenario, a process which can be modeled as 
\begin{equation}\label{trans_model}
	Y_{n_d\times T} = H_{n_d\times n_s}X_{n_s \times T} + N_{n_d \times T},
\end{equation}
where the subscripts $n_s$, $n_d$ and $T$ denote the number of antennas at the source, at the destination, and the number of channel uses, respectively. In the above equation, $Y$ and $X$ are the received and transmitted codewords, $H$ is the random complex \emph{channel matrix} modeling fading, typically assumed to be Rayleigh distributed, and $N$ is a noise matrix whose entries are complex white Gaussian with zero mean and variance $\sigma^2$. We assume that the channel is quasi-static, that is, $H$ stays fixed during the transmission of the whole ST block, and then changes independently of its previous state. The destination is assumed to have perfect channel state information (CSI-D).     

In order to avoid accumulation of the received signals, forcing a discrete (\emph{e.g.}, a lattice) structure on the code is helpful. In this article we will only consider linear ST block codes.

\begin{definition}\label{def:stc}
	Let $\lbrace B_i \rbrace_{i=1}^{k}$ be an independent set of fixed $n_s \times T$ complex matrices. A \emph{linear space--time block code} of rank $k$ is a set of the form
$
	\mathcal{X} = \Biggl\lbrace \sum\limits_{i=1}^{k}{s_i B_i} \Biggm\vert s_i \in S \Biggr\rbrace,
$
where $S \subset \mathbb{Z}$ is the finite \emph{signaling alphabet} used. 
\end{definition}
\begin{definition}\label{def:rate}
The \emph{code rate} of $\mathcal{X}$ is defined as $R = k/T$ real symbol per channel use (rscu), and the code is said to be \emph{full-rate} (for $n_d$ destination antennas) if $k = 2 n_d T$, that is, $R = 2 n_d$.  
\end{definition}

\begin{remark}
In literature, the code rate is commonly defined in complex symbols. However, since we connect the rate to the lattice dimension, it is more convenient to define it over the real alphabet, for not every lattice has a $\mathbb{Z}[i]$-basis. We also want to point out that, here, the channel may be asymmetric $(n_s\neq n_d)$, and hence \emph{full} rate is more meaningfully defined as the maximum rate that still maintains the discrete structure at the receiver and allows for linear detection methods such as sphere-decoding. If the code matrix carries more than  $2 n_d T$ symbols, the received signals will accumulate, and it is thus not desirable to exceed the rate $2n_d$. 
\end{remark}

Henceforth, we will refer to a linear ST block code simply as a ST code, and to its defining matrices $B_i$ as \emph{weight matrices}. Throughout the paper, $S \subset \mathbb{Z}$ will denote the finite signaling alphabet accompanying the considered ST code $\mathcal{X}$, and the superscripts $^\dagger$ and $^T$ the Hermitian conjugate and transpose of a matrix, respectively.

\begin{definition}
A ST code $\mathcal{X}$ as above whose weight matrices $\lbrace B_i \rbrace_{i=1}^{k}$ form a basis of a \emph{lattice} $\Lambda \subset \Mat(n_s\times T, \mathbb{C})$, that is a discrete Abelian subgroup of $\Mat(n_s\times T,\mathbb{C})$, is called a \emph{ST lattice code} of rank $k = \rk(\Lambda)$, $k\leq 2n_sT$. In case of equality, $\Lambda$ is called a full-rank lattice. 
\end{definition}	  
	 
Consider a ST code $\mathcal{X}$ and let $X \neq X'$ denote code matrices ranging over $\mathcal{X}$. We briefly recall the most important design criteria for ensuring a reliable performance:
\begin{itemize}
	\item \emph{Diversity gain:} $\min_{X \neq X'} \rk(X-X') = \min\{n_s,T\}$. A ST code satisfying this criterion is called a \emph{full-diversity} code.

	\item \emph{Coding gain:} $\Delta_{\min} := \min_{X \neq X'} \det[(X-X')(X-X')^{\dagger}]$ should be (after normalization to unit volume, see \cite{hollanti_mindet}) as big as possible. If $\inf \Delta_{\min}  > 0$ for the infinite code 
	\[
	\mathcal{X}_\infty = \left\{ \sum\limits_{i=1}^{k}{s_i B_i} \Biggm\vert s_i \in \mathbb{Z}\right\},
\]
	\emph{i.e.}, the determinants do not vanish when the code size increases, the ST code is said to have the \emph{Nonvanishing Determinant} (NVD) property. 
\end{itemize}

These criteria can be ensured by choosing the algebraic structure underlying the codes in a smart way. Indeed, Algebraic Number Theory and the theory of central simple algebras and their orders have been proven useful for constructing good ST codes (see \cite{hollanti_order1,hollanti_order2,vehkalahti2,oggier_perfect} among others).

\subsection{Space--Time Codes from Division Algebras}
\label{subsec:cda}

For the rest of this paper, we assume that that the number $n_s$ of transmit antennas and the number of channel uses $T$ coincide, unless stated otherwise. Due to this assumption, the delay $T$ grows with the total number of transmit antennas (in the virtual channel, see Section~\ref{subsec:relay_channel}).

Division algebras were first considered in \cite{sethuraman} as a tool for ST coding, leading to fully-diverse codes. The NVD property was first achieved in \cite{belfiore} for the \emph{Golden code}, and the results were generalized to other \emph{Perfect codes} in \cite{oggier_perfect}, where the lattices used for code construction were additionally forced to be orthogonal. The orthogonality requirement was later sacrificed in \cite{hollanti_order1,hollanti_order2,vehkalahti2} for improved performance, and the use of maximal orders was proposed to get denser lattices and higher coding gains. Finally, it was noted in \cite{hollanti_mindet,vehkalahti2} that the comparison of different ST codes requires meaningful normalization. We will now revise some of these notions. 

\begin{proposition}\label{prop:full-diversity}\cite[Prop.~1]{sethuraman}
	Let $\mathbf{F}$ be a field and $\mathcal{D}$ a division $\mathbf{F}$-algebra. Let $\phi: \mathcal{D} \mapsto \Mat(n, \mathbf{F})$ be a ring homomorphism and $\mathcal{X} \subset \phi(\mathcal{D})$ a finite subset. Then, $\rk(X-X') = n$ for any distinct $X, X' \in \mathcal{X}$. 
\end{proposition}

Full-diversity can thus be guaranteed by choosing the underlying algebraic structure to be a division algebra, while imposing a further algebraic restriction will also ensure the NVD property. Among different types of division algebras, \emph{Cyclic Division Algebras} (CDAs) from number field extensions have been proposed in \cite{sethuraman} and heavily used for ST coding ever since.

\begin{definition}\label{def:cda}
	Let $\mathbf{K}/\mathbf{F}$ be a cyclic Galois extension of degree $n$ of number fields, and fix a generator $\sigma$ of its cyclic Galois group $\Gamma(\mathbf{K}/\mathbf{F})$. A \emph{cyclic algebra} of degree $n$ is a triple 
\[	
	\mathcal{C} = (\mathbf{K}/\mathbf{F}, \sigma, \gamma) := \bigoplus\limits_{i=0}^{n-1}{u^i\mathbf{K}},
\]
where $u^n = \gamma \in \mathbf{K}^\times$ and $\kappa u = u \sigma(\kappa)$ for all $\kappa \in \mathbf{K}$. 
The algebra $\mathcal{C}$ is \emph{division}, if every nonzero element of $\mathcal{C}$ is invertible.
\end{definition}

\begin{remark}\label{rmk:quat_algebra}
	If $n = 2$, then $\mathbf{K} = \mathbf{F}(\sqrt{a})$ for some square-free $a \in \mathbb{Z}$. The algebra $\mathcal{C} = (\mathbf{K}/\mathbf{F},\sigma,\gamma)$ is known as a \emph{quaternion algebra}, and can equivalently be denoted as 
$
		\mathcal{C} = (a,\gamma)_{\mathbf{F}} \cong \mathbf{K} \oplus \mathbf{j} \mathbf{K} \cong \mathbf{F} \oplus \mathbf{i}\mathbf{F} \oplus \mathbf{j}\mathbf{F} \oplus \mathbf{k}\mathbf{F}, 
$ 
	where the basis elements satisfy $\mathbf{i}^2 = a$, $\mathbf{j}^2 = \gamma$, $\mathbf{ij} = -\mathbf{ji} = \mathbf{k}$. The case $a = \gamma = -1$ gives rise to the famous \emph{Hamiltonian quaternions} and well-known \emph{Alamouti code}.  
\end{remark}

The following lemmas, the first being a straightforward generalization of an original result due to A. Albert, give us simple ways to determining whether a cyclic algebra $\mathcal{C}$ is division. Denote by $\Nm_{\mathbf{K}/\mathbf{F}}(\cdot)$ the field norm of $\mathbf{K}$ over $\mathbf{F}$.

\begin{lemma}\label{lem:cda}\cite[Prop.~2.4.5]{hollanti_thesis}
	Let $\mathcal{C} = (\mathbf{K}/\mathbf{F},\sigma,\gamma)$ be a cyclic algebra of degree $n$. If $\gamma$ is chosen such that $\gamma^{n/p} \notin \Nm_{\mathbf{K}/\mathbf{F}}(\mathbf{K}^\times)$ for all primes $p \mid n$, then $\mathcal{C}$ is a division algebra. 
\end{lemma} 

\begin{lemma}\label{lem:quat}\cite[Thm.~7.1]{unger}
	Let $\mathbf{F}$ be a number field with ring of integers $\mathcal{O}_{\mathbf{F}}$, and $\mathfrak{p}$ a prime ideal of $\mathcal{O}_{\mathbf{F}}$, with corresponding $\mathfrak{p}$-adic valuation $\nu_{\mathfrak{p}}(\cdot)$. Let $a \in \mathbf{F}$ be such that $\nu_{\mathfrak{p}}(a) = 1$. Then, for any element $\gamma \in \mathcal{O}_{\mathbf{F}}$ which is not a square $\bmod\ \mathfrak{p}$, the quaternion algebra $(a,\gamma)_{\mathbf{F}}$ is division.
\end{lemma}

A further advantage of using CDAs for ST coding is that a lattice structure is easily ensured by restricting the choice of elements to a ring within the CDA known as an \emph{order}.

\begin{definition}
	Let $\mathbf{K}/\mathbf{F}$ be an extension of number fields, $\mathcal{O}_{\mathbf{F}}$ the ring of integers of $\mathbf{F}$, and $\mathcal{C}$ a $\mathbf{K}$-central algebra. 
An $\mathcal{O}_{\mathbf{F}}$-\emph{order} $\mathcal{O}$ in $\mathcal{C}$ is a subring of $\mathcal{C}$ that shares the same identity element as $\mathcal{C}$ and so that $\mathcal{O}\cdot\mathbf{F} = \mathcal{C}$. 
Further, \emph{maximality} is defined with respect to inclusion. 
\end{definition}

We can easily construct a ST code (cf. Def.~\ref{def:stc}) by using the left-regular representation of the CDA. To that end, let $\mathcal{O}$ be an order within the CDA $\mathcal{C} = (\mathbf{K}/\mathbf{F},\sigma,\gamma)$ of degree $n$ and $\alpha = \sum_{i=0}^{n-1}{\alpha_i u^i} \in \mathcal{O}$. The representation of $\alpha$ over the maximal subfield $\mathbf{K}$ is given by   
\begin{equation}\label{eqn:left_reg}
	\lambda: \alpha \mapsto \left[\begin{smallmatrix} 
\alpha_0 & \gamma\sigma(\alpha_{n-1}) & \cdots & \gamma\sigma^{n-1}(\alpha_1) \\
\vdots & \vdots & & \vdots \\
\alpha_{n-1} & \sigma(\alpha_{n-2}) & \cdots & \sigma^{n-1}(\alpha_0)	
	\end{smallmatrix}\right].
\end{equation}
The map $\lambda$ is an isomorphism, allowing us to identify an element $\alpha$ with its matrix representation.

Let $k$ be the absolute degree of $\mathcal{C}$ over $\mathbb{Q}$ and $\lbrace B_i \rbrace_{i=1}^{k}$ a matrix basis of $\mathcal{O}$ over $\mathbb{Q}$. A ST code constructed from the order $\mathcal{O}$ for a fixed signaling alphabet $S \subset \mathbb{Z}$ is of the form
\[
	\mathcal{X} = \Biggl\lbrace \sum\limits_{i=1}^{k}{s_i B_i} \Biggm\vert s_i \in S \Biggr\rbrace. 
\]

By choosing $\mathcal{C}$ to be division, and since $\lambda: \mathcal{O} \to \Mat(n,\mathbf{K})$ is a ring homomorphism, the code $\mathcal{X} \subset \lambda(\mathcal{O})$ is a finite subset of a lattice, and by Prop.~\ref{prop:full-diversity} is fully diverse. Moreover, the restriction of the elements to an $\mathcal{O}_{\mathbf{F}}$-order $\mathcal{O}$ ensures that for any matrix $\lambda(a)$, $\det(\lambda(a)) \in \mathcal{O}_{\mathbf{F}}$, thus guaranteeing the NVD property for $\mathbf{F} = \mathbb{Q}$ or $\mathbf{F}$ imaginary quadratic (cf. \cite{vehkalahti2}).

\subsection{Fast-Decodable Space--Time Codes}
\label{subsec:fd}

The use of multiple antennas for data exchange has many advantages, but also increases the complexity of the coding schemes; especially when decoding the received signal. The considered ST codes as in Def.~\ref{def:stc} allow for a decoding technique known as \emph{Maximum-Likelihood} (ML) decoding, which, given a ST code $\mathcal{X}$ and the transmission model \eqref{trans_model}, and recalling that the noise involved has zero mean, amounts to finding the codeword $X \in \mathcal{X}$ that minimizes 
\begin{equation}
\label{eqn:ml_decod}
\delta(X) := ||Y-HX||_F^2,
\end{equation}
where $||\cdot ||_F$ denotes the Frobenius norm. 

\begin{definition}\label{def:fd}
	The \emph{ML decoding complexity} of a rank-$k$ ST code $\mathcal{X}$ is defined as the minimum number of values that have to be computed for finding the solution to \eqref{eqn:ml_decod}. It is upper bounded by the worst-case (ML) complexity $|S|^k$ corresponding to an exhaustive search. 

A ST code $\mathcal{X}$ is said to be \emph{Fast-Decodable} (FD) if its worst-case ML decoding complexity is of the form\footnote{It is not sufficient to demand $k' < k$, as the ML decoding complexity of \emph{any} code can be reduced to $|S|^{k-2}$ due to Gram-Schmidt orthogonalization.} $|S|^{k'}$ for $k' < k-2$. 
\end{definition}

\begin{remark}
	In the rest of this article, when we say that a code has complexity $|S|^{k'}$, we mean the aforementioned worst-case complexity. By using, e.g., sphere-decoding the complexity can be of course reduced, both for FD and non-FD codes, but the search dimension will be determined by $|S|^{k'}$ as we need to jointly decode $k'$ symbols. Hence $|S|^{k'}$ gives us a way to compare complexities independently of the decoding method we would finally choose to use.
\end{remark}

Although algebraic ST codes can provide big diversity gains and offer high multiplexing gains, and therefore achieve the DMT of specific channels -- not only in theory (DMT is an asymptotic measure) but also in practice -- the major bottle-neck in effective implementation of algebraic ST-codes has traditionally been their (ML) decoding complexity. The aforementioned gains have thus been threatened to remain mainly theoretical, and the concept of fast-decodability was introduced in \cite{biglieri} in order to address the possibility for reducing the dimension of the (ML) decoding problem without having to resort to suboptimal decoding methods.

Let $H$ be the channel matrix and $\vect(\cdot): \Mat(m \times n,\mathbb{C}) \to \mathbb{R}^{2 m n}$ the map which stacks the columns of a matrix followed by separating the real and imaginary parts of the obtained vector components. For a set of weight matrices $\left\{B_i\right\}_{i=1}^k$, define $B :=  \left[\begin{smallmatrix}\vect(HB_1) & \ldots & \vect(HB_k)\end{smallmatrix}\right] \in \Mat(2T n_d\times k, \mathbb{R})$, so that every received codeword $HX$ can be represented as $Bs$ for a coefficient vector $s = (s_1,\ldots,s_k)^T \in S^k$. The problem of decoding now reads
\[
	\underset{X \in \mathcal{X}}{\argmin}\lbrace ||Y-HX||_F^2\rbrace \leadsto \underset{s \in S^k}{\argmin}\lbrace ||\vect(Y) - Bs||_E^2 \rbrace, 
\]
where $||\cdot ||_E$ denotes the Euclidean norm, 
and a real sphere-decoder can be employed to perform the latter search \cite{viterbo}. 
Applying $QR$-decomposition on $B$, where $Q$
$\in \Mat(2n_s n_d\times k,\mathbb{R})$ 
is an orthonormal matrix, and $R \in \Mat(k,\mathbb{R})$ is upper triangular, simplifies the above expression further to finding $s \in S^k$ that minimizes
$
	||\vect(Y)-QRs||_E^2 = ||Q^{\dagger}\vect(Y)-Rs||_E^2.
$

Introducing the $R$-matrix in the decoding process permits to directly read out the decoding complexity of a given code. To that end, we will use a specific quadratic form which has been introduced in \cite{jithamithra} as a tool for studying fast-decodability independently of the channel matrix.

\begin{definition}\label{def:hrqf}
The \emph{Hurwitz-Radon Quadratic Form} (HRQF) is the map
\[
		\mathcal{Q}: \mathcal{X} \to \mathbb{R}; \qquad
		X \mapsto \sum\limits_{1 \le i \le j \le k}{s_i s_j m_{ij}},
\]
where $m_{ij} := ||B_i B_j^{\dagger} + B_j B_i^{\dagger}||_{F}^2$ and $s_i \in S$.
Associating the matrix $M = (m_{ij})_{i,j}$ with $\mathcal{Q}$, the HRQF can be written as $\mathcal{Q}(X) = sMs^{T}$, where $s = \left[\begin{smallmatrix} s_1 & \ldots & s_k \end{smallmatrix}\right]$.
\end{definition}

\begin{remark}
\label{rmk:mutual}
	Note that $m_{ij} = ||B_i B_j^{\dagger} + B_j B_i^{\dagger}||_F^2 = 0$ if and only if $B_i B_j^{\dagger} + B_j B_i^{\dagger} = \mathbf{0}$, that is, if $B_i$ and $B_j$ are \emph{mutually orthogonal}. Moreover, premultiplication of the weight matrices by $H$ does not affect the zero structure of $M$, whereas it does affect that of $R$. Yet, the zero structure of the $R$ and $M$ matrices are conveniently related to each other (see Prop.~\ref{prop:cond_g_decod}, Cor.~\ref{cor:g_decod} or \cite{jithamithra}). 
\end{remark}

We further specify two important families of FD codes, together with an explicit decoding complexity expression. 

\begin{itemize}[leftmargin=5mm]
	\item[1.]Conditional $g$-group decodability:
\begin{definition}\label{def:cond_g_decod}
	A ST code $\mathcal{X}$ is \emph{conditionally $g$-Group Decodable} ($g$-GD) if there exists a partition of $\lbrace 1,\ldots,k \rbrace$ into $g+1$ nonempty subsets $\Gamma_1,\ldots,\Gamma_g,\Gamma^{\mathcal{X}}$, $g \ge 2$, such that $B_i B_j^{\dagger} + B_j B_i^{\dagger} = \mathbf{0}$ for $i \in \Gamma_u$, $j \in \Gamma_v$ and $1 \le u < v \le g$. 
\end{definition}

\begin{proposition}\label{prop:cond_g_decod}\cite[Thm.~2]{berhuy}
	There exists an ordering of the weight matrices\footnote{An algorithm for finding the optimal ordering is given in \cite{jithamithra}.} such that the $R$-matrix obtained for conditionally $g$-GD ST codes has the particular form
	\[
		R = \left[\begin{smallmatrix} D & N' \\ & N \end{smallmatrix}\right] = \left[\begin{smallmatrix} D_1 & & & N_1 \\ & \ddots & & \vdots \\ & & D_g & N_g \\ & & & N \end{smallmatrix}\right], 
	\]
	where $D$ is a $(k-|\Gamma^{\mathcal{X}}|)\times(k-|\Gamma^{\mathcal{X}}|)$ block-diagonal matrix whose blocks $D_i$ are of size $|\Gamma_i|\times|\Gamma_i|$, $N$ is a square upper-triangular $|\Gamma^{\mathcal{X}}|\times|\Gamma^{\mathcal{X}}|$ matrix, and $N'$ is a rectangular matrix. 
\end{proposition}

\begin{remark}
	We intentionally refer to \cite{berhuy} for this result although the authors do not use the term \emph{conditional $g$-group decodable}. However, Def.~\ref{def:cond_g_decod} above coincides with the codes considered in the referred source. A similar observation was also already made in \cite{srinath}.
\end{remark}

	\item[2.]$g$-group decodability:
\begin{definition}\label{def:g_decod}
	$\mathcal{X}$ is \emph{$g$-group decodable} if there exists a partition of $\lbrace 1,\ldots,k\rbrace$ into $g$ nonempty subsets $\Gamma_1,\ldots,\Gamma_g$ such that $B_i B_j^{\dagger} + B_j B_i^{\dagger} = \mathbf{0}$ for $i \in \Gamma_u$, $j \in \Gamma_v$, and $u \neq v$. 
\end{definition}

\begin{remark}
A code $\mathcal{X}$ is $g$-GD if it is conditionally $g$-GD and its associated set $\Gamma^{\mathcal{X}}$ is empty.   
\end{remark}

From this remark and Prop.~\ref{prop:cond_g_decod} we immediately get the following corollary. 

\begin{corollary}\label{cor:g_decod}
	There exists an ordering of the weight matrices such that the $R$-matrix obtained for $g$-GD ST codes is of the form 
	\[
		R = \left[\begin{smallmatrix} D_1 & & \\ & \ddots & \\ & & D_g \end{smallmatrix}\right],
	\]
	where $D_i$ is a $|\Gamma_i|\times|\Gamma_i|$ upper-triangular matrix.
\end{corollary}
\end{itemize}

By Def.~\ref{def:fd}, both conditionally $g$-GD and $g$-GD ST codes are FD. The latter definitions, however, allow to deduce the exact decoding complexity reduction with the help of the HRQF. For this purpose, note that having a (conditionally) $g$-GD ST code, decoding the last $|\Gamma^{\mathcal{X}}| \ge 0$ variables gives a complexity of $|S|^{|\Gamma^{\mathcal{X}}|}$. The remaining variables can be decoded in $g$ parallel steps, where step $i$ involves $|\Gamma_i|$ variables. This observation leads to the following result. 

\begin{proposition}\label{prop:complexity}
	Given a (conditionally) $g$-GD ST code $\mathcal{X}$ with possibly empty subset $\Gamma^{\mathcal{X}}$, the decoding complexity of $\mathcal{X}$ is 
	\[
		|S|^{|\Gamma^{\mathcal{X}}| + \max\limits_{1 \le i \le g}{|\Gamma_i|}}.
	\]
\end{proposition}

\subsection{Iterative Construction from Cyclic Algebras}
\label{subsec:iterated_cons}

Crucial for ST codes to exhibit desirable properties is the underlying algebraic framework. Constructing codes for larger number of antennas means dealing with higher degree field extensions and algebras, which are harder to handle. We briefly recall an iterative ST code construction, recently proposed in \cite{markin}, which, starting with an $n\times n$ ST code, results in a new $2n\times 2n$ ST code with the same rate and double rank. The advantage of this construction is that when applied carefully, the resulting codes inherit good properties from the original ST codes. 

\begin{definition}\label{def:iterated_map}
Let $\mathbf{F}$ be a finite Galois extension of $\mathbb{Q}$ and $\mathcal{C} = (\mathbf{K}/\mathbf{F},\sigma,\gamma)$ be a CDA of degree $n$. Fix $\theta \in \mathcal{C}$ and $\tau \in \Aut_{\mathbb{Q}}(\mathbf{K})$, \emph{i.e.}, a $\mathbb{Q}$-automorphism of $\mathbf{K}$. 
\begin{itemize}
	\item[(a)] Define the function
	\[
		\begin{split}
			\alpha_{\tau,\theta}: \Mat(n,\mathbf{K})\times\Mat(n,\mathbf{K}) &\to \Mat(2n,\mathbf{K}) \\
			(X,Y) &\mapsto \left[\begin{smallmatrix} X & \theta\tau(Y) \\ Y & \tau(X) \end{smallmatrix}\right].
		\end{split}
	\]
	
	\item[(b)] If $\theta = \zeta\theta'$ is totally real or totally imaginary, define the alike function	
	\[
		\begin{split}
			\tilde{\alpha}_{\tau,\theta}: \Mat(n,\mathbf{K})\times\Mat(n,\mathbf{K}) &\to \Mat(2n,\mathbf{K}) \\
			(X,Y) &\mapsto \left[\begin{smallmatrix} X & \zeta\sqrt{\theta'}\tau(Y) \\ \sqrt{\theta'}Y & \tau(X) \end{smallmatrix}\right].
		\end{split}	
	\]
\end{itemize}	
\end{definition}

Suppose that $\mathcal{C}$ gives rise to a ST code $\mathcal{X}$ of rank $k$ defined via matrices $\lbrace B_i \rbrace_{i=1}^{k}$. Then, the matrices $\left\{ \alpha_{\tau,\theta}(B_i,0), \alpha_{\tau,\theta}(0,B_i)\right\}_{i=1}^{k}$ (or applying $\tilde{\alpha}_{\tau,\theta}(\cdot,\cdot)$, respectively) define a rank-$2k$ code 
\[
	\mathcal{X}_{\text{it}} = \Biggl\lbrace \sum\limits_{i=1}^{k}\left[s_i\alpha_{\tau,\theta}(B_i,0) + s_{k+i}\alpha_{\tau,\theta}(0,B_i)\right] \Biggm\vert s_i \in S \Biggr\rbrace. 	
\] 

\begin{proposition}\label{prop:iterated_construction} \cite[Thm.~1, Thm.~2]{markin}
	Let $\mathcal{C} = (\mathbf{K}/\mathbf{F}, \sigma, \gamma)$ be a CDA giving rise to a ST code $\mathcal{X}$ defined by the matrices $\lbrace B_i \rbrace_{i=1}^{k}$. Assume that $\tau \in \Aut_{\mathbb{Q}}(\mathbf{K})$ commutes with $\sigma$ and complex conjugation, and further $\tau(\gamma) = \gamma$, $\tau^2 = \id$. Fix $\theta \in \mathbf{F}^{\langle \tau \rangle}$, where $\mathbf{F}^{\langle \tau \rangle}$ is the subfield of $\mathbf{F}$ fixed by $\tau$. Identifying an element of $\mathcal{C}$ with its matrix representation (cf. \eqref{eqn:left_reg}), we have:
	\begin{itemize}
		\item[(i)] The image $\mathcal{I} = \alpha_{\tau,\theta}(\mathcal{C},\mathcal{C})$ is an algebra and is division if and only if $\theta \neq z\tau(z)$ for all $z \in \mathcal{C}$. Moreover, for any $\alpha_{\tau,\theta}(x,y) \in \mathcal{I}$, we have $\det(\alpha_{\tau,\theta}(x,y)) \in \mathbf{F}^{\langle \tau \rangle}$. 
		
		\item[(ii)] If in addition $\theta = \zeta\theta'$ is totally real or totally imaginary, the image $\tilde{\mathcal{I}} = \tilde{\alpha}_{\theta}(\mathcal{C},\mathcal{C})$ retains both the full-diversity and NVD property. If for some $i,j$, $B_i B_j^{\dagger} + B_j B_i^{\dagger} = \mathbf{0}$, we have
		\[
			\begin{split}
				\tilde{\alpha}_{\tau,\theta}(B_i,0)\tilde{\alpha}_{\tau,\theta}(B_j,0)^{\dagger} + \tilde{\alpha}_{\tau,\theta}(B_j,0)\tilde{\alpha}_{\tau,\theta}(B_i,0)^{\dagger} &= \mathbf{0}\,, \\
				\tilde{\alpha}_{\tau,\theta}(0,B_i)\tilde{\alpha}_{\tau,\theta}(0,B_j)^{\dagger} + \tilde{\alpha}_{\tau,\theta}(0,B_j)\tilde{\alpha}_{\tau,\theta}(0,B_i)^{\dagger} &= \mathbf{0}.
			\end{split}
		\]
	\end{itemize}
\end{proposition}

\section{Fast-Decodable Space--Time Codes for Distributed Communications}
\label{sec:distributed}

We consider the communication of $N+1$ users with a single destination, where every user as well as the destination can be equipped with multiple antennas. In this scenario, enabling cooperation permits that the active transmitter be assisted by the other $N$ users in the transmission of its data. We start by introducing the assumed channel model in detail.  

\subsection{MIMO Relay Channel}
\label{subsec:relay_channel}

In the following, we consider $N+1$ users communicating to a single destination over a wireless network. Each user is allocated a time slot for the transmission of its data, that is the channel is shared in a time-division multiple-access manner. Within a fixed time slot, the remaining $N$ users act as intermediate relays, helping the active transmitter in the communication process by amplifying and forwarding the received signal. Considering a single time slot, the channel resembles a single user channel with $N$ relays and a single destination, as illustrated in Figure~\ref{fig:systemfig}. 
\begin{figure}[h]
\begin{small}
\centering
\begin{minipage}[b]{0.38\textwidth}
\begin{tikzpicture}[node distance=1cm]
 \node[relay] (r1) {\scriptsize Relay 1};
 \node[below=0.15cm of r1] (vert) {\scriptsize $\textbf \vdots$};
 \node[relay, below=0.15cm of vert] (rN) {\scriptsize Relay N};
 \node[above=0.5cm of r1] (dummy) {};
 \node[user, right=1.75cm of dummy] (dest) {\scriptsize Dest.}
   edge[pil_rev] node[pos=0.75, above]{\tiny $G_1$} (r1.east) 
   edge[pil_rev] node[pos=0.65, left]{\tiny $G_{N}$} (rN.east); 
 \node[user, left=1.75cm of dummy] (source) {\scriptsize Source}
   edge[pil] node[above]{\scriptsize $F$} (dest.west)       
   edge[pil_dash] node[pos=0.75, above]{\tiny $H_1$} (r1.west)
   edge[pil_dash] node[pos=0.65, right]{\tiny $H_{N}$} (rN.west);
  \node[below=0.75cm of rN] (dummy2) {};
\end{tikzpicture}
\caption{System model for the $N$-relay channel with a single destination.}
\label{fig:systemfig}
\end{minipage}\hfill
\begin{minipage}[b]{0.6\textwidth}
\begin{tikzpicture}
[blanknode/.style={rectangle, draw=white, minimum width=10mm},
solidnode/.style={rectangle, draw=black, minimum width=10mm, minimum height=2mm},
dashednode/.style={rectangle, draw=black, dashed, minimum width=10mm, minimum height=2mm},]

\node[blanknode]		(source)	                    		{Source};
\node[blanknode]		(r1)		[below=0.1cm of source]	{$R_1$};
\node[blanknode]    	(r2)    	[below=0.1cm of r1] 		{$R_2$};
\node[blanknode]    	(dummyr)	[below=0.1cm of r2] 		{$\vdots$};
\node[blanknode]    	(rN)    	[below=0.1cm of dummyr] 	{$R_N$};
\node[blanknode]		(dest)	[below=0.1cm of rN]		{Dest.};

\node[solidnode]		(x11)	[right=0.2cm of source]	{\footnotesize{$X_{1,1}$}};
\node[solidnode]		(x12)	[right=0.05cm of x11]	{\footnotesize{$X_{1,2}$}};
\node[solidnode]		(x21)	[right=0.05cm of x12]	{\footnotesize{$X_{2,1}$}};
\node[solidnode]		(x22)	[right=0.05cm of x21]	{\footnotesize{$X_{2,2}$}};
\node[blanknode]		(dummyt)[right=0.05cm of x22]	{\footnotesize{$\cdots$}};
\node[solidnode]		(xN1)	[right=0.05cm of dummyt]	{\footnotesize{$X_{N,1}$}};
\node[solidnode]		(xN2)	[right=0.05cm of xN1]	{\footnotesize{$X_{N,2}$}};

\node[dashednode]	(r11)	[below=0.05cm of x11]	{\footnotesize{$X_{1,1}$}};
\node[solidnode]		(r12)	[below=0.05cm of x12]	{\footnotesize{$X_{1,1}$}};

\node[dashednode] 	(r21)	[below=0.65cm of x21]	{\footnotesize{$X_{2,1}$}};
\node[solidnode] 	(r22)	[below=0.65cm of x22]	{\footnotesize{$X_{2,1}$}};

\node[blanknode]		(dummyd)[below=1.25cm of dummyt]	{\footnotesize{$\ddots$}};

\node[dashednode]	(rN1)	[below=2.15cm of xN1]	{\footnotesize{$X_{N,1}$}};
\node[solidnode]		(rN2)	[below=2.15cm of xN2]	{\footnotesize{$X_{N,1}$}};

\node[dashednode]	(d11) 	[below=2.8cm of x11]	{\footnotesize{$Y_{1,1}$}}
	edge node[pos=0.5, below=0.3cm]{$0$} (d11.west); 
\node[dashednode]	(d12) 	[below=2.8cm of x12]	{\footnotesize{$Y_{1,2}$}}
	edge node[pos=0.5, below=0.3cm]{$\frac{T}{2}$} (d11.east); 
\node[dashednode]	(d21) 	[below=2.8cm of x21]	{\footnotesize{$Y_{2,1}$}}
	edge node[pos=0.5, below=0.3cm]{$T$} (d12.east); 
\node[dashednode]	(d22) 	[below=2.8cm of x22]	{\footnotesize{$Y_{2,2}$}}
	edge node[pos=0.5, below=0.3cm]{$\frac{3T}{2}$} (d21.east)
	edge node[pos=0.5, below=0.3cm]{$2T$} (d22.east); 
\node[blanknode]		(dummyc)	[below=2.8cm of dummyt]	{\footnotesize{$\cdots$}};
\node[dashednode]	(dN1) 	[below=2.8cm of xN1]	{\footnotesize{$Y_{N,1}$}};
\node[dashednode]	(dN2) 	[below=2.8cm of xN2]	{\footnotesize{$Y_{N,2}$}}
	edge node[pos=0.5, below=0.3cm]{$NT$} (dN2.east); 
\end{tikzpicture}
\caption{Superframe structure for the $N$-relay NAF channel. Transmitted and received signals are represented by solid and dashed boxes, respectively.}
\label{fig:frame}
\end{minipage}
\end{small}
\end{figure}
The matrices $F$, $H_i$ and $G_i$, $1 \le i \le N$ denote the Rayleigh distributed channels from the source to the destination, relays, and from the relays to the destination, respectively. 

Henceforth, we will assume the Nonorthogonal Amplify-and-Forward (NAF) scheme introduced in \cite{nabar} and generalized in \cite{yang} to the MIMO case, where, in contrast to the orthogonal schemes, the active transmitter and helping relay can transmit at the same time. In addition, we assume the \emph{half-duplex} constraint, \emph{i.e.}, the relays can only receive or transmit a signal at a given time instance. 
For a fixed time slot, let us define a superframe consisting of $N$ consecutive cooperation frames, during which the relays take turns to cooperate with the active transmitter in their respective cooperation frame. Each frame of length $T$ is composed of two partitions of $T/2$ symbols, and all channels are assumed to be static during the transmission of the entire superframe. This frame model is depicted in Figure~\ref{fig:frame}. 
Denote by $n_s$, $n_d$ and $n_r$ the number of antennas at the source, destination, and each of the $N$ relays, respectively. In the following, we assume $n_r \le n_s$. In the case $n_r > n_s$, the relays can do better than simply forwarding the received signal, and we refer to \cite{yang} for a brief discussion of the possible strategies. 
As illustrated in Figure~\ref{fig:systemfig} and \ref{fig:frame}, the transmission process can be modeled as
\begin{align*}
		Y_{i,1} &= \sqrt{\pi_1\SNR} F X_{i,1} + V_{i,1}\,, &i=1,\ldots,N \\
		Y_{i,2} &= \sqrt{\pi_2\SNR} F X_{i,2} + V_{i,2} + \sqrt{\pi_3\SNR	} G_i B_i(\sqrt{\pi_1\rho\SNR} H_i X_{i,1} + W_i)\,, &i=1,\ldots,N
\end{align*}
where $Y_{i,j}$, $X_{i,j}$ are the received and transmitted matrices, $V_{i,j}$, $W_i$ represent additive white Gaussian noise, the matrices $B_i$ are used for normalization and $\pi_i$ are channel-independent, chosen so that $\SNR$ denotes the received Signal-to-Noise Ratio per antenna at the destination. The ratio between the path loss of the source-relay and source-destination links is denoted by $\rho$.  
 
From the destination's point of view, the above transmission model can equivalently be presented as a virtual single-user MIMO channel modeled as 
\[
	Y_{n_d\times n} = H_{n_d\times n} X_{n\times n} + V_{n_d\times n},
\]
where $n=N(n_s+n_r)$, $X$ and $Y$ are the (overall) transmitted and received signals, and the structure of the channel matrix $H$ is determined by the different relay paths. Thus this virtual antenna array created by allowing cooperation can be used to exploit spatial diversity even when a local antenna array may not be available. We have also made the assumption that $T=n$.  

It was shown in \cite{yang} that given a rate-$4n_s$ (cf. Def. \ref{def:rate}) block-diagonal ST code $\mathcal{X}$, that is where each $X \in \mathcal{X}$ takes the form 
\[
	X = \diag\lbrace\Xi_i\rbrace_{i=1}^{N} = \left[\begin{smallmatrix} \Xi_1 & & \\ & \ddots & \\ & & \Xi_N \end{smallmatrix}\right]
\] 
with $\Xi_i \in \Mat(2n_s, \mathbb{C})$ and such that $\mathcal{X}$ is NVD, the equivalent code 
\[
\begin{split}
	C = \left[\begin{smallmatrix} C_1 & \cdots & C_{N-1}\end{smallmatrix}\right], \text{ where } 
	C_i = \left[\begin{smallmatrix} \Xi_i\left[1:n_s, 1:2n_s\right] & \Xi_i\left[n_s+1:2n_s,1:2n_s\right]\end{smallmatrix}\right] 
\end{split}
\]
achieves the optimal DMT for the channel\footnote{DMT is an asymptotic performance measure, indicating a tradeoff between the transmission rate and decoding error probability as a function of SNR. Here, our primary goal is fast-decodability with NVD, so instead of giving a detailed definition of the DMT, we refer the reader to \cite{zheng},\cite{yang}.}, transmitting $C_i$ in the $i^{\text{th}}$ cooperation frame.  

It would thus be desirable to have ST codes of this block-diagonal form which achieve:
\begin{enumerate}
	\item Full rate $2n_d$, that is, the number of independent real symbols (\emph{e.g.}, Pulse Amplitude Modulation (PAM)) per codeword equals $2n_dN(n_s+n_r)$. 
	
	\item Full rank $N(n_s+n_r)$.
	
	\item NVD.
\end{enumerate}

In addition, the aim of this article is to show how such ST codes can be constructed which moreover are FD, thus exhibit a reduced complexity in decoding without sacrificing performance.

\subsection{Fast-Decodable Distributed Space--Time Codes}
\label{subsec:relay_codes}

In the following, we present the main results of this article, that is we introduce flexible methods for constructing FD ST codes for the Single-Input Multiple-Output (SIMO) and MIMO-NAF channel for $N \ge 1$ relays. The following function is crucial for the proposed constructions.
\begin{definition}\label{def:coop_construction}
Consider an $N$-relay NAF channel. Given a ST code $\mathcal{X} \subset \Mat(n_s+n_r,\mathbb{C})$ and a suitable function $\eta$ of order $N$ (\emph{i.e.}, $\eta^N(X)=X$), define the function
\[
		\Psi_{\eta,N}: \mathcal{X} \to \Mat(nN,\mathbb{C}); \quad
		X \mapsto \diag\lbrace \eta^i(X)\rbrace_{i=0}^{N-1} = \left[\begin{smallmatrix} X & & \\ & \ddots & \\ & & \eta^{N-1}(X) \end{smallmatrix}\right].
\]
\end{definition}

For the remaining of this section, we use techniques from Algebraic Number Theory to prove certain properties about the structures used for code construction. An interested reader is referred to \cite{marcus} as a good source to review these techniques.
We denote by $\imath := \sqrt{-1}$ the complex unit. 

\subsubsection{\underline{SIMO}}
\label{ssubsec:single}

Assume a cooperative communications scenario as illustrated in Figure~\ref{fig:systemfig}, where $n_s = n_r = 1$, $n_d \ge 2$. Consider the tower of extensions depicted in Figure~\ref{fig:siso_tower}, where $\xi$ is taken to be totally real, $m \in \mathbb{Z}_{\ge 1}$ and $a \in \mathbb{Z}\backslash\left\{0\right\}$ are square-free.
\begin{figure}[h]
\begin{small}
\centering
\begin{tikzpicture}[node distance=1cm]
 \node (alg) {$\mathcal{C} = (a,\gamma)_{\mathbf{K}} \cong (\mathbf{L}/\mathbf{K},\sigma:\sqrt{a} \mapsto -\sqrt{a},\gamma)$};
 \node[below=0.5cm of alg] (L) {$\mathbf{L} = \mathbf{K}(\sqrt{a})$}
 	edge[-] node[pos=0.5, right]{\scriptsize 2} (alg);
 \node[below=0.5cm of L] (K) {$\mathbf{K} = \mathbf{F}(\xi)$}
 	edge[-] node[pos=0.5, right]{\scriptsize 2} (L);
 \node[below=0.5cm of K] (F) {$\mathbf{F} = \mathbb{Q}(\sqrt{-m})$}
 	edge[-] node[pos=0.5, right]{\scriptsize N} (K);
 \node[below=0.5cm of F] (Q) {$\mathbb{Q}$}
 	edge[-] node[pos=0.5, right]{\scriptsize 2} (F);
 \node[left=0.5cm of F] (Qi) {$\mathbb{Q}(\sqrt{a})$}
 	edge[-] node[pos=0.5, above left]{\scriptsize 2N} (L)
 	edge[-] node[pos=0.5, below left]{\scriptsize 2} (Q);
\end{tikzpicture}
\caption{Tower of extensions for the SISO code construction.}
\label{fig:siso_tower}
\end{small}
\end{figure}
Assume that $\mathcal{C}$ is division (cf. Lemma~\ref{lem:cda}). 
Let $ \sigma$ be the generator of $ \Gamma(\mathbf{L}/\mathbf{K})$, and fix a generator $ \eta $ of $ \Gamma(\mathbf{K}/\mathbf{F})$. 

 To have balanced energy and good decodability, it is necessary to slightly modify the matrix representation of the elements in $\mathcal{C}$, thus for $c,d \in \mathcal{O}_{\mathbf{L}}$, $\mathcal{O} \subset \mathcal{C}$ an order, instead of representing $x = c+\sqrt{\gamma} d \in \mathcal{O}$ by its left-regular representation $\lambda(x)$, we define the following similar and commonly used function that maintains the original determinant,
\begin{equation}\label{eqn:lambda_sim}
	\tilde{\lambda}: x \mapsto \left[\begin{smallmatrix} c & -\sqrt{-\gamma}\sigma(d) \\ \sqrt{-\gamma}d & \sigma(c) \end{smallmatrix}\right].
\end{equation}

\begin{theorem}\label{thm:single_antenna_relay_code}
Arising from the algebraic setup in Figure~\ref{fig:siso_tower} with $a < 0$, $\gamma < 0$, define the set
\[
		\mathcal{X} = \lbrace \Psi_{\eta,N}(X)\rbrace_{X \in \tilde{\lambda}(\mathcal{O})} = \Bigl\lbrace \diag\lbrace \eta^i(X) \rbrace_{i=0}^{N-1} \Bigm\vert X \in \tilde{\lambda}(\mathcal{O}) \Bigr\rbrace.
\]
	
The code $\mathcal{X}$ is of rank $8N$, rate $R=4$ rscu and has the NVD property. It is full-rate if $n_d = 2$. Moreover, $\mathcal{X}$ is conditionally $4$-GD, and its decoding complexity can be reduced from $|S|^{8N}$ to $|S|^{5N}$, where $S$ is the real constellation used, resulting in a complexity reduction of $37.5\%$.
\end{theorem}

\begin{proof}
	Let $\lbrace \beta_1 = 1, \ldots , \beta_{N} \rbrace$ be an $\mathbf{F}$-basis of $\mathbf{K}$ and consider the following $\mathbf{K}$-basis of $\mathcal{C}$:
\[
		\left\{
		\Gamma_{1,1} = \left[\begin{smallmatrix} 1 & 0 \\ 0 & 1 \end{smallmatrix}\right], 
		\Gamma_{2,1} = \left[\begin{smallmatrix} \sqrt{a} & 0 \\ 0 & -\sqrt{a} \end{smallmatrix}\right], \Gamma_{3,1} = \left[\begin{smallmatrix} 0 & \sqrt{-\gamma}\sqrt{a} \\ \sqrt{-\gamma}\sqrt{a} & 0  \end{smallmatrix}\right], 
		\Gamma_{4,1} = \left[\begin{smallmatrix} 0 & -\sqrt{-\gamma} \\ \sqrt{-\gamma} & 0 \end{smallmatrix}\right] \right\}.
\]
We first extend this basis to an $\mathbf{F}$-basis of $\mathcal{C}$ as $\left\{ \Gamma_{i,j} = \Gamma_{i,1}\beta_j \right\}_{\substack{1 \le i \le 4 \\ 1 \le j \le N}}$, and further to a $\mathbb{Q}$-basis of $\mathcal{C}$ by complementing with $\Gamma_{i,j} = \left\{\sqrt{-m}\Gamma_{i-4,j}\right\}_{\substack{5 \le i \le 8 \\ 1 \le j \le N}}$. We get a $\mathbb{Z}$-basis for the code $\mathcal{X}$ as 
\[
	\mathcal{B} = \lbrace \Psi_{\eta,N}(\Gamma_{i,j})\rbrace_{\substack{1 \le i \le 8 \\ 1 \le j \le N}},
\]
which is of length $8N = \rk(\mathcal{X})$. Thus by Def.~\ref{def:rate}, the code has rate $R = 4$. 

Let now $\Psi_{\eta,N}(X) \in \mathcal{X}$ be a codeword, where $\eta$ denotes a generator of $\Gamma(\mathbf{K}/\mathbf{F})$. As assumed, the coefficients of $X$ are taken from the ring of integers $\mathcal{O}_{\mathbf{L}}$ of $\mathbf{L}$, thus $\det(X) \in \mathcal{O}_{\mathbf{K}}$ and hence
\[
		\det(\Psi_{\eta,N}(X)) = \prod\limits_{i=0}^{N-1}{\det(\eta^i(X))} = \prod\limits_{i=0}^{N-1}{\eta^i(\det(X))} 
		= \Nm_{\mathbf{K}/\mathbf{F}}(\det(X)) \in \mathcal{O}_{\mathbf{F}}.  
\] 
Since $\mathbf{F}$ is imaginary quadratic, it follows that $\det(\Psi_{\eta,N}(X)) \ge 1$, giving the NVD property. 

For deriving the decoding complexity, a direct computation shows that $\Gamma_{i,j}\Gamma_{u,v}^{\dagger} + \Gamma_{u,v}\Gamma_{i,j}^{\dagger} = \mathbf{0}$ for $1 \le i,u\le 4$, $i \neq u$ and $1 \le j,v \le N$, hence $\Psi_{\eta,N}(\Gamma_{i,j})\Psi_{\eta,N}(\Gamma_{u,v})^{\dagger} + \Psi_{\eta,N}(\Gamma_{u,v})\Psi_{\eta,N}(\Gamma_{i,j})^{\dagger} = \mathbf{0}$. 
Consequently, the matrix $M = (m_{ij})$ associated with the HRQF is of the form
\[
	M = \left[\begin{smallmatrix}
	\ast & 0 & 0 & 0 & \ast & \ast & \ast & \ast \\
	0 & \ast & 0 & 0 & \ast & \ast & \ast & \ast \\
	0 & 0 & \ast & 0 & \ast & \ast & \ast & \ast \\
	0 & 0 & 0 & \ast & \ast & \ast & \ast & \ast \\
	\ast & \ast & \ast & \ast & \ast & \ast & \ast & \ast \\
	\ast & \ast & \ast & \ast & \ast & \ast & \ast & \ast \\
	\ast & \ast & \ast & \ast & \ast & \ast & \ast & \ast \\
	\ast & \ast & \ast & \ast & \ast & \ast & \ast & \ast 
	\end{smallmatrix}\right]
\]
where each of the entries is an $N\times N$ matrix and is the zero matrix, if the corresponding entry is $0$. By Def.~\ref{def:cond_g_decod}, $\mathcal{X}$ is conditionally $4$-GD and exhibits decoding complexity $|S|^{4N+N} = |S|^{5N}$. 
\end{proof}

\begin{remark}
	ST codes constructed using the method from Theorem~\ref{thm:single_antenna_relay_code} will achieve the optimal DMT of the channel, according to \cite[Thm. 4]{yang}. 
\end{remark}

\begin{figure}[h]
\begin{small}
	\begin{minipage}{0.50\textwidth}
\centering
\begin{tikzpicture}[node distance=1cm]
 \node (alg) {$\mathcal{C} = (-3,-\frac{2}{\sqrt{5}})_{\mathbf{K}} \cong (\mathbf{L}/\mathbf{K},\sigma:\sqrt{-3} \mapsto -\sqrt{-3},-\frac{2}{\sqrt{5}})$};
 \node[below=0.5cm of alg] (L) {$\mathbf{L} = \mathbf{K}(\sqrt{-3})$}
 	edge[-] node[pos=0.5, right]{\scriptsize 2} (alg);
 \node[below=0.5cm of L] (K) {$\mathbf{K} = \mathbb{Q}(\imath,\xi)$}
 	edge[-] node[pos=0.5, right]{\scriptsize 2} (L);
 \node[below=0.5cm of K] (F) {$\mathbf{F} = \mathbb{Q}(\imath)$}
 	edge[-] node[pos=0.5, right]{\scriptsize 2} (K);
 \node[below=0.5cm of F] (Q) {$\mathbb{Q}$}
 	edge[-] node[pos=0.5, right]{\scriptsize 2} (F);
 \node[left=1cm of F] (Qi) {$\mathbb{Q}(\sqrt{-3})$}
 	edge[-] node[pos=0.5, above left]{\scriptsize 4} (L)
 	edge[-] node[pos=0.5, below left]{\scriptsize 2} (Q);
\end{tikzpicture}
\caption{Tower of extensions for a 2-relay SISO example code.}
\label{fig:siso_example1}
	\end{minipage}\hfill
	\begin{minipage}{0.49\textwidth}
\centering
\begin{tikzpicture}[node distance=1cm]
 \node (alg) {$\mathcal{C} = (-11,-1)_{\mathbf{K}} \cong (\mathbf{L}/\mathbf{K},\sigma:\sqrt{-11} \mapsto -\sqrt{-11},-1)$};
 \node[below=0.5cm of alg] (L) {$\mathbf{L} = \mathbb{Q}(\sqrt{-11},\zeta_7)$}
 	edge[-] node[pos=0.5, right]{\scriptsize 2} (alg);
 \node[below=0.5cm of L] (K) {$\mathbf{K} = \mathbb{Q}(\zeta_7)$}
 	edge[-] node[pos=0.5, right]{\scriptsize 2} (L);
 \node[below=0.5cm of K] (F) {$\mathbf{F} = \mathbb{Q}(\sqrt{-7})$}
 	edge[-] node[pos=0.5, right]{\scriptsize 3} (K);
 \node[below=0.5cm of F] (Q) {$\mathbb{Q}$}
 	edge[-] node[pos=0.5, right]{\scriptsize 2} (F);
 \node[left=0.5cm of F] (Qi) {$\mathbb{Q}(\sqrt{-11})$}
 	edge[-] node[pos=0.5, above left]{\scriptsize 6} (L)
 	edge[-] node[pos=0.5, below left]{\scriptsize 2} (Q);
\end{tikzpicture}
\caption{Tower of extensions for a 3-relay SISO example code.}
\label{fig:siso_example2}
	\end{minipage}
\end{small}
\end{figure}

\begin{example}
\label{exp:simo_fd}
For $N = 2$ relays and $\xi = \sqrt{5}$, consider the tower of extensions in Figure~\ref{fig:siso_example1}. 
The algebra $\mathcal{C}$ is division. To see this, note that $\mathfrak{q} = 3\mathcal{O}_{\mathbf{K}}$ be an ideal which splits as a product $\mathfrak{q} = \mathfrak{p}_1\mathfrak{p}_2$ of prime ideals $\mathfrak{p}_1$, $\mathfrak{p}_2$ of $\mathcal{O}_{\mathbf{K}}$, with residue class degree $\mathfrak{f}_{\mathfrak{p}_i} = 2$. Let $\mathfrak{p}$ be any of these prime ideals and consider its corresponding $\mathfrak{p}$-adic valuation $\nu_{\mathfrak{p}}(\cdot)$. We have $\nu_{\mathfrak{p}}(-3) = 1$, and further, $\mathcal{O}_{\mathbf{K}}/\mathfrak{p} \cong \mathbb{F}_{3^2}$. 
Since $5 \equiv 2 \bmod\ 3$, clearly $\sqrt{5}\notin \mathbb{F}_3$. Also, we have $-\frac{\sqrt{5}}{2} = -2\sqrt{5} = \sqrt{5}$ in characteristic $3$. It follows $\mathbb{F}_{3^2} = \mathbb{F}_3(\sqrt{5})$, and clearly $\sqrt[4]{5} \notin \mathbb{F}_{3^2}$. Hence, $-\frac{\sqrt{5}}{2}$ and thus $-\frac{2}{\sqrt{5}}$ is not a square $\bmod\ \mathfrak{p}$. By Lemma \ref{lem:quat}, $\mathcal{C}$ is division. 
 
Let $x = c + \sqrt{-\gamma}d$ with $c,d \in \mathcal{O}_{\mathbf{L}}$ and for $\sigma$ as above, 
$
	X = \tilde{\lambda}(x) = \left[\begin{smallmatrix} c & -\sqrt{-\gamma}\sigma(d) \\ \sqrt{-\gamma}d & \sigma(c) \end{smallmatrix}\right].
$
For $\langle \eta \rangle = \Gamma(\mathbf{K}/\mathbf{F})$, define the $2$-relay code
\[
	\begin{split}
	\mathcal{X} = \left\{ \Psi_{\eta,2}(X) \right\}_{X \in \tilde{\lambda}(\mathcal{O}_{\mathbf{L}})} 
	= \left\{\left. \diag\lbrace\eta^i(X)\rbrace_{i=0}^{1} = \left[\begin{smallmatrix} X & \\ & \eta(X) \end{smallmatrix}\right] \right| X \in \tilde{\lambda}(\mathcal{O}_{\mathbf{L}}) \right\}. 
	\end{split}
\]

The resulting code is a fully diverse NVD code of rank 16, which is conditionally $4$-GD having decoding complexity $|S|^{10}$ in contrast to $|S|^{16}$.  
\end{example}

\begin{example}
\label{exp:single_antenna_fd2}
Let $N = 3$ be the number of relays and $\zeta_7$ the $7^{\text{th}}$ root of unity, and consider the tower of extensions in Figure~\ref{fig:siso_example2}. Note that $\mathbf{K} = \mathbf{F}(\xi)$, where $\xi = \zeta_7+\zeta_7^{-1}$ is totally real. 

The algebra $\mathcal{C}$ is division. To see this, note that the ideal $\mathfrak{q} = (-11)\mathcal{O}_{\mathbf{K}}$ splits as $\mathfrak{q} = \mathfrak{p}_1\mathfrak{p}_2$ for distinct prime ideals $\mathfrak{p}_i$ of $\mathcal{O}_{\mathbf{K}}$ with residue class degree $\mathfrak{f}_{\mathfrak{p}_i}$. Let $\mathfrak{p}$ be either of these prime ideals, with corresponding $\mathfrak{p}$-adic valuation $\nu_{\mathfrak{p}}(\cdot)$. We have that $\mathcal{O}_{\mathbf{K}}/\mathfrak{p} \cong \mathbb{F}_{11^3}$. Next we establish that $-1$ is not a square in $\mathbb{F}_{11^3}$, which follows from the fact that $\ord(-1) = 2$ and $4 \nmid | \mathbb{F}_{11^3}^\times |$.  Hence, $-1 $ is not a square $\bmod\ \mathfrak{p}$. By Lemma \ref{lem:quat}, $\mathcal{C}$ is division.

Note that $\mathcal{C} \cong (-1,-11)_{\mathbf{K}}$ and let $x = c+\sqrt{-11}d$ with $c,d \in \mathbb{Z}[\imath,\zeta_7]$, $\sigma: \imath \mapsto -\imath$, so that 
$
	X = \tilde{\lambda}(x) = \left[\begin{smallmatrix} c & -\sqrt{11}\sigma(d) \\ \sqrt{11}d & \sigma(c) \end{smallmatrix}\right].
$ 
For $\langle \eta:\zeta_7 \mapsto \zeta_7^2 \rangle = \Gamma(\mathbf{K}/\mathbf{F})$, define the 3-relay code
\[
		\mathcal{X} = \lbrace \Psi_{\eta,3}(X) \rbrace_{X \in \tilde{\lambda}(\mathcal{O}_{\mathbf{L}})}
		= \Biggl\lbrace \diag\lbrace \eta^i(X)\rbrace_{i=0}^{2} = \left[\begin{smallmatrix} X & & \\ & \eta(X) & \\ & & \eta^2(X) \end{smallmatrix}\right] \Biggm\vert X \in \tilde{\lambda}(\mathcal{O}_{\mathbf{L}}) \Biggr\rbrace.
\] 

The constructed full-diversity code $\mathcal{X}$ is of rank 24, satisfies the NVD property and has decoding complexity $|S|^{15}$ in contrast to $|S|^{24}$.  
\end{example}

\subsubsection{\underline{MIMO}}
\label{ssubsec:multiple}

In the following, the single source is now equipped with $n_s \ge 1$ antennas, and for $p \ge 5$ prime, let $N  = (p-1)/2$ be the number of relays, each equipped with $n_r \ge 1$ antennas and such that\footnote{This assumption allows to construct codes from quaternion algebras, whereas the case $n_s + n_r > 4$ requires working with larger cyclic division algebras. While it is straightforward to achieve full-diversity and NVD, general statements about the exact decoding complexity reduction become harder.} $n_r + n_s = 4$. Assume further a single destination with $n_d \ge 1$ antennas, and consider the tower of extensions in Figure~\ref{fig:mimo_tower},
\begin{figure}[h]
\begin{small}
\centering
\begin{tikzpicture}[node distance=1cm]
 \node (alg) {$\mathcal{C} = (a,\gamma)_{\mathbf{K}} \cong (\mathbf{L}/\mathbf{K},\sigma:\sqrt{a} \mapsto -\sqrt{a},\gamma)$};
 \node[below=0.5cm of alg] (L) {$\mathbf{L} = \mathbf{K}(\sqrt{a})$}
 	edge[-] node[pos=0.5, right]{\scriptsize 2} (alg);
 \node[below=0.5cm of L] (K) {$\mathbf{K} = \mathbb{Q}(\xi)$}
 	edge[-] node[pos=0.5, right]{\scriptsize 2} (L);
 \node[below=0.5cm of K] (dummy) {};
 \node[below=0.5cm of dummy] (Q) {$\mathbb{Q}$}
 	edge[-] node[pos=0.5, right]{\scriptsize N} (K);
 \node[left=1cm of dummy] (Qi) {$\mathbf{F} = \mathbb{Q}(\sqrt{a})$}
 	edge[-] node[pos=0.5, above left]{\scriptsize N} (L)
 	edge[-] node[pos=0.5, below left]{\scriptsize 2} (Q);
\end{tikzpicture}
\caption{Tower of extensions for the MIMO code construction.}
\label{fig:mimo_tower}
\end{small}
\end{figure}
where $\mathbf{K} = \mathbb{Q}(\xi) = \mathbb{Q}^+(\zeta_p) \subset \mathbb{Q}(\zeta_p)$ is the maximal real subfield of the $p^{\text{th}}$ cyclotomic field, that is, $\xi = \zeta_p+\zeta_p^{-1}$, and $a \in \mathbb{Z}\backslash\left\{0\right\}$ is square-free. Let $\langle \sigma \rangle = \Gamma(\mathbf{L}/\mathbf{K})$ and $\langle \eta \rangle = \Gamma(\mathbf{L}/\mathbf{F})$. 

\begin{theorem}
\label{thm:mult_antenna_relay_code}
	In the setup as in Figure~\ref{fig:mimo_tower}, choose $a \in \mathbb{Z}_{< 0}$ such that $\mathfrak{p} = a\mathcal{O}_{\mathbf{K}}$ is a prime ideal. Fix further $\gamma < 0$ and $\theta \in \mathcal{O}_{\mathbf{K}}\cap \mathbb{R}^\times = \mathbb{Z}[\xi]\cap\mathbb{R}^\times$ such that 
	\begin{itemize}
		\item $\gamma$ and $\theta$ are both nonsquare $\bmod\ \mathfrak{p}$,
		\item the quadratic form $\langle\gamma,-\theta\rangle_{\mathbf{L}}$ is anisotropic, that is evaluates to zero if and only if $\gamma = \theta = 0$,
	\end{itemize}
and further let $\tau = \sigma$. Then, if $\mathcal{O}\subset \mathcal{C}$ is an order, the distributed ST code
\[
	\mathcal{X} = \left\{\left. \Psi_{\eta,N}(\tilde{\alpha}_{\tau,\theta}(X,Y)) = \diag\left\{\eta^i(\tilde{\alpha}_{\tau,\theta}(X,Y))\right\}_{i=0}^{N-1} \right| X,Y \in \tilde{\lambda}(\mathcal{O})\right\}
\]
is a full-diversity ST code of rank $8N$, rate $R = 2$ rscu, exhibits the NVD property and is FD. Its decoding complexity is $|S|^{k'}$, where $k' = 4N$ if $a \equiv 1 \bmod\ 4$, and $k' = 2N$ if $a \not\equiv 1 \bmod\ 4$, resulting in a reduction in complexity of $50\%$ and $75\%$, respectively.
\end{theorem}

\begin{proof}
	We first show that $\mathcal{X}$ is fully diverse, wherefore it is enough to show that $\mathcal{C}$ is division. By Lemma~\ref{lem:cda}, it suffices to show that $\gamma \notin \Nm_{\mathbf{L}/\mathbf{K}}(\mathbf{L}^{\times})$. Let $\alpha = \alpha_0+\sqrt{a}\alpha_1 \in \mathbf{L^{\times}}$. Then $\Nm_{\mathbf{L}/\mathbf{K}}(\alpha) = \alpha\sigma(\alpha) = \alpha_0^2-a\alpha_1^2$. Thus $\gamma = \Nm_{\mathbf{L}/\mathbf{K}}(\mathbf{L}^{\times}) \Leftrightarrow \alpha_0^2-a\alpha_1^2-\gamma = 0$ has nontrivial solutions in $\mathbf{K}$. But $a < 0$, $\gamma < 0$ and $\mathbf{K}$ is totally real, thus there can't be any solutions. 
	
	Let now $\left\{b_i\right\}_{i=1}^{N} = \left\{1,\xi,\ldots,\xi^{N-1}\right\}$ be a power basis of $\mathcal{O}_{\mathbf{K}}$ and consider 
\[	
	X = \tilde{\lambda}(x) = \left[\begin{smallmatrix} x_1+x_2\omega & -\sqrt{-\gamma}(x_3+x_4\sigma(\omega)) \\ \sqrt{-\gamma}(x_3+x_4\omega) & x_1+x_2\sigma(\omega) \end{smallmatrix}\right],
\]
where $\omega = \sqrt{a}$ if $a \not\equiv 1 \bmod\ 4$ and $\omega = \frac{1+\sqrt{a}}{2}$ otherwise, so that $\mathcal{O}_{\mathbf{F}} = \mathbb{Z}[\omega]$. Let 
\[
	\mathcal{X}_0 = \left\{\left. \sum\limits_{i=1}^{k}{s_i B_i^0} \right| s_i \in S \right\} \subset \tilde{\lambda}(\mathcal{O}),
\]
for a set of matrices $\mathcal{B}_0 = \left\{B_i^0\right\}_{i=1}^{k} = \left\{X(b_i,0,0,0),\ldots,X(0,0,0,b_i)\right\}_{i=1}^{N}$, thus $k = 4N$. Originating from this code, we construct a set of weight matrices defining the iterated code $\mathcal{X}_0^{\text{it}} = \left\{\left. \tilde{\alpha}_{\tau,\theta}(X,Y)\right| X,Y \in \mathcal{X}_0 \right\}$ as $\mathcal{B}_0^{\text{it}} = \left\{B_i^{\text{it}}\right\}_{i=1}^{8N} = \left\{\tilde{\alpha}_{\tau,\theta}(B_i^0,0),\tilde{\alpha}_{\tau,\theta}(0,B_i^0)\right\}_{i=1}^{4N}$, and get from $\mathcal{X}_0^{\text{it}}$ a defining set of weight matrices for the distributed code $\mathcal{X}$ as $\mathcal{B} = \left\{B_i\right\}_{i=1}^{8N} = \left\{\Psi_{\eta,N}(B_i^{\text{it}})\right\}_{i=1}^{8N}$.
The code $\mathcal{X}$ is thus of length $8N$ and by Def.~\ref{def:rate} has rate $R = 8N/4N = 2$ rscu. 

To see that $\mathcal{X}$ is NVD, note that by the restrictions imposed on the entries of elements in $\mathcal{X}_0$, we have $\det(\tilde{\alpha}_{\tau,\theta}(X,Y)) \in \mathcal{O}_{\mathbf{L}}$. Hence
\[
\begin{split}
	\det[\Psi_{\eta,N}(\tilde{\alpha}_{\tau,\theta}(X,Y))] &= \prod\limits_{i=0}^{N-1}{\det[\eta^i(\tilde{\alpha}_{\tau,\theta}(X,Y))]} = \prod\limits_{i=0}^{N-1}{\eta^i[\det(\tilde{\alpha}_{\tau,\theta}(X,Y))]} \\
	&= \Nm_{\mathbf{L}/\mathbf{F}}\left[\det(\tilde{\alpha}_{\tau,\theta}(X,Y))\right] \in \mathcal{O}_{\mathbf{F}}.
\end{split}
\]
As $\mathbf{F}$ is imaginary quadratic, $\det[\Psi_{\eta,N}(\tilde{\alpha}_{\tau,\theta}(X,Y))] \ge 1$. 

It remains to show that $\mathcal{X}$ is FD. To that end, group the matrices $\left\{B_i^0\right\}_{i=1}^{N}$ as follows: 
\[
\begin{split}
	G_1 = \left\{X(b_i,0,0,0)\right\}_{i=1}^{N};\quad G_2 = \left\{X(0,b_i,0,0)\right\}_{i=1}^{N}; \\ 
	G_3 = \left\{X(0,0,b_i,0)\right\}_{i=1}^{N};\quad G_4 = \left\{X(0,0,0,b_i)\right\}_{i=1}^{N}.
\end{split}
\]
Let $X_i \in G_i$. We have for $i = 1,2$ and $j = 3,4$ that $X_i X_j^{\dagger} + X_j X_i^{\dagger} = \mathbf{0}$, and moreover, 
\[
	X_1 X_2^\dagger + X_2 X_1^{\dagger} \begin{cases} = \mathbf{0} &\mbox{if } a \not\equiv 1 \bmod\ 4, \\ \neq \mathbf{0} &\mbox{if } a \equiv 1 \bmod\ 4, \end{cases}
\]
and the same holds for $X_3,X_4$. We thus conclude that $\mathcal{X}$ is $2$-GD if $a \equiv 1 \bmod\ 4$, exhibiting a decoding complexity of $|S|^{2N}$ and is $4$-GD otherwise, in which case its decoding complexity is $|S|^{N}$. By Prop.~\ref{prop:iterated_construction} (ii) and since $\theta$, $\tau$ are chosen to satisfy the requirements of the iterative construction (\cite[Cor.~8]{markin}), the iterated code $\mathcal{X}_0^{\text{it}}$ and consequently the distributed ST code $\mathcal{X}$ exhibit a decoding complexity of $|S|^{4N}$ in the former, and $|S|^{2N}$ in the latter case. 
\end{proof}

\begin{remark}
	We remark that by the results obtained in \cite{berhuy}, the decoding complexity of codes arising from division algebras can only be reduced by a factor of 4. The constructive method proposed in Theorem~\ref{thm:mult_antenna_relay_code} results in codes whose decoding complexity is reduced by either 50\% or 75\% compared to non-FD codes of the same rank. Thus, in the latter case, our codes indeed achieve the maximal complexity reduction by a factor of 4.  
\end{remark}

\begin{example}
\label{exp:mimo_fd}
We construct two codes for $N = 3$ relays, arising from the towers of extensions depicted in Figure~\ref{fig:mimo_example1}, where $\xi = \zeta_7+\zeta_7^{-1}$, $\gamma_1 = -1$, $\gamma_2 = -\frac{2}{1+\xi}$. 
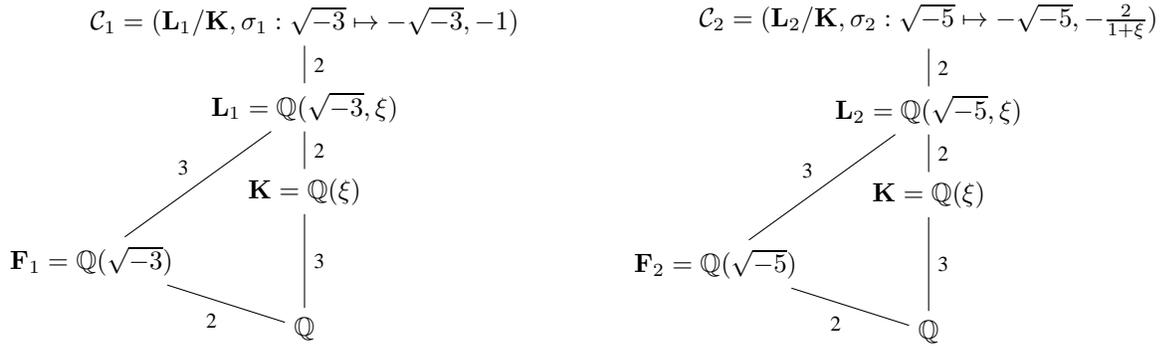
\begin{figure}[h]
\begin{small}
\centering
\begin{minipage}{0.49\textwidth}
\begin{tikzpicture}[node distance=1cm]
 \node (alg) {$\mathcal{C}_1 = (\mathbf{L}_1/\mathbf{K},\sigma_1:\sqrt{-3} \mapsto -\sqrt{-3}, -1)$};
 \node[below=0.5cm of alg] (L) {$\mathbf{L}_1 = \mathbb{Q}(\sqrt{-3},\xi)$}
 	edge[-] node[pos=0.5, right]{\scriptsize 2} (alg);
 \node[below=0.5cm of L] (K) {$\mathbf{K} = \mathbb{Q}(\xi)$}
 	edge[-] node[pos=0.5, right]{\scriptsize 2} (L);
 \node[below=0.5cm of K] (dummy) {};
 \node[below=0.5cm of dummy] (Q) {$\mathbb{Q}$}
 	edge[-] node[pos=0.5, right]{\scriptsize 3} (K);
 \node[left=1.5cm of dummy] (Qi) {$\mathbf{F}_1 = \mathbb{Q}(\sqrt{-3})$}
 	edge[-] node[pos=0.5, above left]{\scriptsize 3} (L)
 	edge[-] node[pos=0.5, below left]{\scriptsize 2} (Q);
\end{tikzpicture}
\end{minipage}\hfill
\begin{minipage}{0.49\textwidth}
\begin{tikzpicture}[node distance=1cm]
 \node (alg) {$\mathcal{C}_2 = (\mathbf{L}_2/\mathbf{K},\sigma_2:\sqrt{-5} \mapsto -\sqrt{-5}, -\frac{2}{1+\xi})$};
 \node[below=0.5cm of alg] (L) {$\mathbf{L}_2 = \mathbb{Q}(\sqrt{-5},\xi)$}
 	edge[-] node[pos=0.5, right]{\scriptsize 2} (alg);
 \node[below=0.5cm of L] (K) {$\mathbf{K} = \mathbb{Q}(\xi)$}
 	edge[-] node[pos=0.5, right]{\scriptsize 2} (L);
 \node[below=0.5cm of K] (dummy) {};
 \node[below=0.5cm of dummy] (Q) {$\mathbb{Q}$}
 	edge[-] node[pos=0.5, right]{\scriptsize 3} (K);
 \node[left=1.5cm of dummy] (Qi) {$\mathbf{F}_2 = \mathbb{Q}(\sqrt{-5})$}
 	edge[-] node[pos=0.5, above left]{\scriptsize 3} (L)
 	edge[-] node[pos=0.5, below left]{\scriptsize 2} (Q);
\end{tikzpicture}
\end{minipage}
\end{small}
\caption{Towers of extensions for two 3-relay MIMO example codes.}
\label{fig:mimo_example1}
\end{figure} 
\end{example}
In the following, for $i = 1,2$, let $\tau_i = \sigma_i$, $\langle\eta_i:\xi\mapsto\xi^2-2\rangle = \Gamma\left(\mathbf{L}_i/\mathbf{F}_i\right)$. Note that $\eta_1 \neq \eta_2$, since they have distinct fixed fields. Choose further $\theta_1 = 1-\xi = \zeta\theta_1'$, $\theta_2 = 3(1-\xi) = \zeta\theta_2'$, with $\zeta = -1$, $\theta_1', \theta_2' \in \mathbb{R}_{>0}$, and let $p_{\min}(x,\xi)$ be the minimal polynomial of $\xi$. 

Since $\mathfrak{p}_1 = (-3)\mathcal{O}_\mathbf{K}$ and $\mathfrak{p}_2 = (-5)\mathcal{O}_\mathbf{K}$ are prime, their residue class degree is $\mathfrak{f}_{\mathfrak{p}_i} = 3$, thus 
\[
	\mathcal{O}_{\mathbf{K}}/\mathfrak{p}_1 \cong \mathbb{F}_{3^3} \cong \mathbb{F}_{3}[x]/p_{\min}(x,\xi); \quad 
	\mathcal{O}_{\mathbf{K}}/\mathfrak{p}_2 \cong \mathbb{F}_{5^3}\cong \mathbb{F}_{5}[x]/p_{\min}(x,\xi). 
\]
We have $\ord(\gamma_1) = 2$ in $\mathcal{O}_{\mathbf{K}}/\mathfrak{p}_1$, and $\ord(\gamma_2) = 124$ in $\mathcal{O}_{\mathbf{K}}/\mathfrak{p}_2$. Since $4 \nmid |\mathbb{F}_{3^3}^\times|$, we establish that $\gamma_i$ is nonsquare $\bmod\ \mathfrak{p}_i$.
Further, $\ord(\theta_1) = 26$ in $\mathcal{O}_{\mathbf{K}}/\mathfrak{p}_1$, and $\ord(\theta_2) = 124$ in $\mathcal{O}_{\mathbf{K}}/\mathfrak{p}_2$. Hence, $\theta_i$ is not a square $\bmod\ \mathfrak{p}_i$. 
Moreover, the quadratic forms $\langle \gamma_i,-\theta_i \rangle_{\mathbf{L}_i}$ are anisotropic, as 
  \[
  	\begin{split}
  	-v_1^2-v_2^2\theta_1 = 0 \text{ for } v_1, v_2 \in \mathbf{L}_1 &\Leftrightarrow \xi-1 = v^2 \text{ for } v \in \mathbf{L}_1 \\
  	-\frac{2}{1+\xi} w_1^2-v_2^2\theta_2 = 0 \text{ for } w_1, w_2 \in \mathbf{L}_2 &\Leftrightarrow \frac{3}{2}(\xi^2-1) = w^2 \text{ for } w \in \mathbf{L}_2,
	\end{split}  
  \]
Since $\xi-1$, $\frac{3}{2}(\xi^2-1)$ are nonsquare in neither field, the conditions from Theorem~\ref{thm:mult_antenna_relay_code} are satisfied. 

Let $x_i \in \mathcal{O}_i \subset \mathcal{C}_i$, and set $\omega_1 = \frac{1+\sqrt{-3}}{2}$, $\omega_2 = \sqrt{-5}$. We define ST codes $\mathcal{X}_{0,i}$ consisting of codewords of the form 
$	X_i = \tilde{\lambda}(x_i) = \left[\begin{smallmatrix} x_{1,i} + x_{2,i}\omega_i & -\sqrt{-\gamma_i}(x_{3,i} + x_{4,i}\sigma_i(\omega_i)) \\ \sqrt{-\gamma_i}(x_{3,i} + x_{4,i}\omega_i) & x_{1,i} + x_{2,i}\sigma_i(\omega_i) \end{smallmatrix}\right],
$
where $x_{j,i} \in \mathcal{O}_{\mathbf{K}}$, $1 \le j \le 4$. Next, we iterate $\mathcal{X}_{0,i}$ to obtain the sets 
\[
	\mathcal{X}_{0,i}^{\text{it}} = \left\{\left.\tilde{\alpha}_{\tau_i,\theta_i}(X,Y) = \left[\begin{smallmatrix} X & \zeta\sqrt{\theta_i'}\tau_i(Y) \\ \sqrt{\theta_i'}Y & \tau_i(X) \end{smallmatrix}\right] \right| X,Y \in \tilde{\lambda}(\mathcal{O}_i) \right\}, 
\]
and finally adapt the two iterated codes to the 3-relay channel by applying the maps $\eta_i$, resulting in distributed ST codes
\[
	\mathcal{X}_i = \left\{\left.\Psi_{\eta_i,3}(\tilde{\alpha}_{\tau_i,\theta_i}(X,Y)) = \diag\left\{\eta_i^j(\tilde{\alpha}_{\tau_i,\theta_i}(X,Y))\right\}_{0 \le j \le 2} \right| X,Y \in \tilde{\lambda}(\mathcal{O}_i)\right\},
\]

Both resulting relay codes are fully diverse, exhibit the NVD property and are FD. While $\mathcal{X}_1$ is $2$-GD with decoding complexity $|S|^{12}$ as opposed to $|S|^{24}$, $\mathcal{X}_2$ is $4$-GD with decoding complexity $|S|^{6}$ in contrast to $|S|^{24}$, resulting in a complexity reduction of $50\%$ and $75\%$, respectively. 

\begin{figure}[h]
\begin{small}
\begin{minipage}[b]{0.46\textwidth}
\begin{tikzpicture}[node distance=1cm]
 \node (alg) {$\mathcal{C} = (-3,-1)_{\mathbf{K}} \cong (\mathbf{L}/\mathbf{K},\sigma:\sqrt{-3} \mapsto -\sqrt{-3}, -1)$};
 \node[below=0.5cm of alg] (L) {$\mathbf{L} = \mathbb{Q}(\sqrt{-3},\xi)$}
 	edge[-] node[pos=0.5, right]{\scriptsize 2} (alg);
 \node[below=0.5cm of L] (K) {$\mathbf{K} = \mathbb{Q}(\xi)$}
 	edge[-] node[pos=0.5, right]{\scriptsize 2} (L);
 \node[below=0.5cm of K] (dummy) {};
 \node[below=0.5cm of dummy] (Q) {$\mathbb{Q}$}
 	edge[-] node[pos=0.5, right]{\scriptsize 5} (K);
 \node[left=1cm of dummy] (Qi) {$\mathbf{F} = \mathbb{Q}(\sqrt{-3})$}
 	edge[-] node[pos=0.5, above left]{\scriptsize 5} (L)
 	edge[-] node[pos=0.5, below left]{\scriptsize 2} (Q);
\end{tikzpicture}
\caption{Tower of extensions for a 5-relay MIMO code.}
\label{fig:mimo_example2}
\end{minipage}\hfill
\begin{minipage}[b]{0.52\textwidth}
\begin{tikzpicture}[node distance=1cm]
 \node (alg) {$\mathcal{C} = (-3,1-\zeta_5)_{\mathbf{K}} \cong (\mathbf{L}/\mathbf{K},\sigma:\sqrt{-3} \mapsto -\sqrt{-3}, 1-\zeta_5)$};
 \node[below=0.5cm of alg] (L) {$\mathbf{L} = \mathbb{Q}(\zeta_5,\sqrt{-3})$}
 	edge[-] node[pos=0.5, right]{\scriptsize 2} (alg);
 \node[below=0.5cm of L] (K) {$\mathbf{K} = \mathbb{Q}(\zeta_5)$}
 	edge[-] node[pos=0.5, right]{\scriptsize 2} (L);
 \node[below=0.5cm of K] (dummy) {};
 \node[below=0.5cm of dummy] (Q) {$\mathbb{Q}$}
 	edge[-] node[pos=0.5, right]{\scriptsize 4} (K);
 \node[left=1cm of dummy] (Qi) {$\mathbf{F} = \mathbb{Q}(\sqrt{-3})$}
 	edge[-] node[pos=0.5, above left]{\scriptsize 4} (L)
 	edge[-] node[pos=0.5, below left]{\scriptsize 2} (Q);
\end{tikzpicture}
\caption{Tower of extensions for a 4-relay MIMO code.}
\label{fig:mimo_counterexample}
\end{minipage}
\end{small}
\end{figure}

\begin{example}
We construct a FD distributed ST code for $N = 5$ relays, arising from the tower of extensions in Figure~\ref{fig:mimo_example2}, where $\xi = \zeta_{11}+\zeta_{11}^{-1}$. Let $\tau = \sigma$ and $\langle \eta:\xi\mapsto\xi^2-2 \rangle = \Gamma(\mathbf{L}/\mathbf{F})$. Choose $\theta = 1-\xi = \zeta\theta'$ with $\zeta = -1$ and $\theta' \in \mathbb{R}_{>0}$. Note that $\mathfrak{p} = (-3)\mathcal{O}_{\mathbf{K}}$ is a prime ideal, and further $\ord(\gamma) = 2$, $\ord(\theta) = 242$ in $\mathcal{O}_{\mathbf{K}}/\mathfrak{p}$. Since there is no element of order $4$ in $\mathcal{O}_{\mathbf{K}}/\mathfrak{p}$, both elements are nonsquare $\bmod\ \mathfrak{p}$. Moreover, the quadratic form $\langle \gamma,-\theta\rangle_{\mathbf{L}}$ is anisotropic, as 
\[
-v_1^2 - v_2^2\theta = 0 \text{ for } v_1,v_2 \in \mathbf{L} \Leftrightarrow -\theta = v^2
\] 
for some $v \in \mathbf{L}$. But $-\theta$ is not a square in $\mathbf{L}$, thus the conditions from Theorem~\ref{thm:mult_antenna_relay_code} are satisfied. 

Let $x\in \mathcal{O} \subset \mathcal{C}$, $\omega = \frac{1+\sqrt{-3}}{2}$, and for $x_1,\ldots,x_4 \in \mathcal{O}_{\mathbf{K}}$, we construct a ST code $\mathcal{X}_0$ consisting of codewords of the form 
$
	X = \tilde{\lambda}(x) = \left[\begin{smallmatrix} x_1+x_2\omega & -(x_3+x_4\sigma(\omega)) \\ x_3+x_4\omega & x_1+x_2\sigma(\omega) \end{smallmatrix}\right].
$

To adapt this code to the proposed scenario, we first iterate $\mathcal{X}_0$ to obtain the set
\[
	\mathcal{X}_0^{\text{it}} = \left\{\left.\tilde{\alpha}_{\tau,\theta}(X,Y) = \left[\begin{smallmatrix} X & \zeta\sqrt{\theta'}\tau(Y) \\ \sqrt{\theta'}Y & \tau(X) \end{smallmatrix}\right] \right| X,Y \in \tilde{\lambda}(\mathcal{O}) \right\}, 
\]
and finally make use of the map $\eta$ to construct the distributed ST code 
\[
	\mathcal{X} = \left\{\left.\Psi_{\eta,5}(\tilde{\alpha}_{\tau,\theta}(X,Y)) = \diag\left\{\eta^i(\tilde{\alpha}_{\tau,\theta}(X,Y))\right\}_{0 \le i \le 4}
	\right| X,Y \in \tilde{\lambda}(\mathcal{O})\right\}.
\]

The resulting relay code is fully diverse and moreover NVD. It is FD and more specifically $2$-GD with decoding complexity $|S|^{20}$ as opposed to $|S|^{40}$, thus resulting in a reduction of 50\%. 
\end{example}

\begin{example}
We conclude this section with an example for $N = 4$ relays that demonstrates the importance of the conditions in Theorem~\ref{thm:mult_antenna_relay_code}. Consider the algebraic setup in Figure~\ref{fig:mimo_counterexample}, where $\zeta_5$ is the $5^{\text{th}}$ root of unity and $\xi = \zeta_5 + \zeta_5^{-1}$. Let $\tau = \sigma$ and $\langle \eta \rangle = \Gamma(\mathbf{L}/\mathbf{F})$. Choose further $\theta = \frac{\zeta_5+1}{\zeta_5-1}$. The quaternion algebra $\mathcal{C}$ is division, and the choice of $\tau$ and $\theta$ satisfy the criteria required in Theorem~\ref{thm:mult_antenna_relay_code}. 
To see this, note that $\mathfrak{p} = (-3)\mathcal{O}_{\mathbf{K}}$ is a prime ideal with residue field $\mathcal{O}_{\mathbf{K}}/\mathfrak{p} \cong \mathbb{F}_{3^4}$. The order of $\gamma$ and $\theta$ within the multiplicative group $\mathbb{F}_{3^4}^{\times}$ are $\ord(\gamma) = 80$, $\ord(\theta) = 16$. Since there is no element of order $32$ in $\mathbb{F}_{3^4}^{\times}$, they are both nonsquare $\bmod\ \mathfrak{p}$. 

Further, the quadratic form $\langle \gamma, -\theta \rangle_{\mathbf{L}}$ is anisotropic. This is as for $v_1, v_2 \in \mathbf{L}$,
\[
	v_1^2\gamma - v_2^2\theta = 0 \Leftrightarrow v^2 = \frac{1+\zeta_5}{(1-\zeta_5)(\zeta_5-1)} = -\frac{1}{5}\alpha 
\]
with $\alpha = 3\zeta_5^3+4\zeta_5^2+3\zeta_5$ and $v \in \mathbf{L}$. It thus suffices to show that $\alpha$ is not a square in $\mathbf{L}$. But it is $\alpha = p_1 p_2^2$, for $p_1, p_2$ prime.  

Let $x = c+\sqrt{\gamma}d$, $c,d \in \mathcal{O}_{\mathbf{L}}$ and for 
$
	X = \lambda(x) = \left[\begin{smallmatrix} c & \gamma\sigma(d) \\ d & \sigma(c) \end{smallmatrix}\right],
$
define the distributed ST code
\[
		\mathcal{X} = \left\{\left. \Psi_{\eta,4}(\alpha_{\tau,\theta}(X,Y))		= \diag\left\{\eta^i(\alpha_{\tau,\theta}(X,Y))\right\}_{i=0}^{3}
		 \right| X,Y \in \lambda(\mathcal{O}) \right\}.
\]

The choices of $\gamma$ and $\mathbf{K}$ do not agree with Theorem~\ref{thm:mult_antenna_relay_code}, and the constructed ST code is in fact not FD. The resulting code exhibits a decoding complexity of $|S|^{30}$ as opposed to $|S|^{32}$, where the reduction is merely due to the Gram-Schmidt orthogonalization. 
\end{example}

\subsection{Simulation Results}
\label{subsec:simulations}
The construction methods proposed in the previous section facilitate the design of ST codes for the N-relay SIMO and MIMO channel which, in addition to being fully diverse and having the NVD property, are FD. The goal of this section is to disclose the actual performance of explicit codes constructed with the proposed methods. 

We start by comparing the performance of the FD code constructed in Example~\ref{exp:simo_fd} with the optimal ST code proposed in \cite[Section~VI-B]{yang}, a lifted version of the Golden code, for two relays with a single antenna. We fix the number of antennas at the destination to be $n_d = 4$. In addition, we also illustrate the performance of an unscaled version of the lifted Golden code, whose entries, as opposed to the scaled version, are only restricted to lie in $\mathcal{O}_{\mathbf{L}} = \mathbb{Z}\left[\zeta_8,\frac{1+\sqrt{5}}{2}\right]$. 
\begin{figure}[h]
\centering
	\includegraphics[width=0.45\textwidth, height=0.3\textwidth]{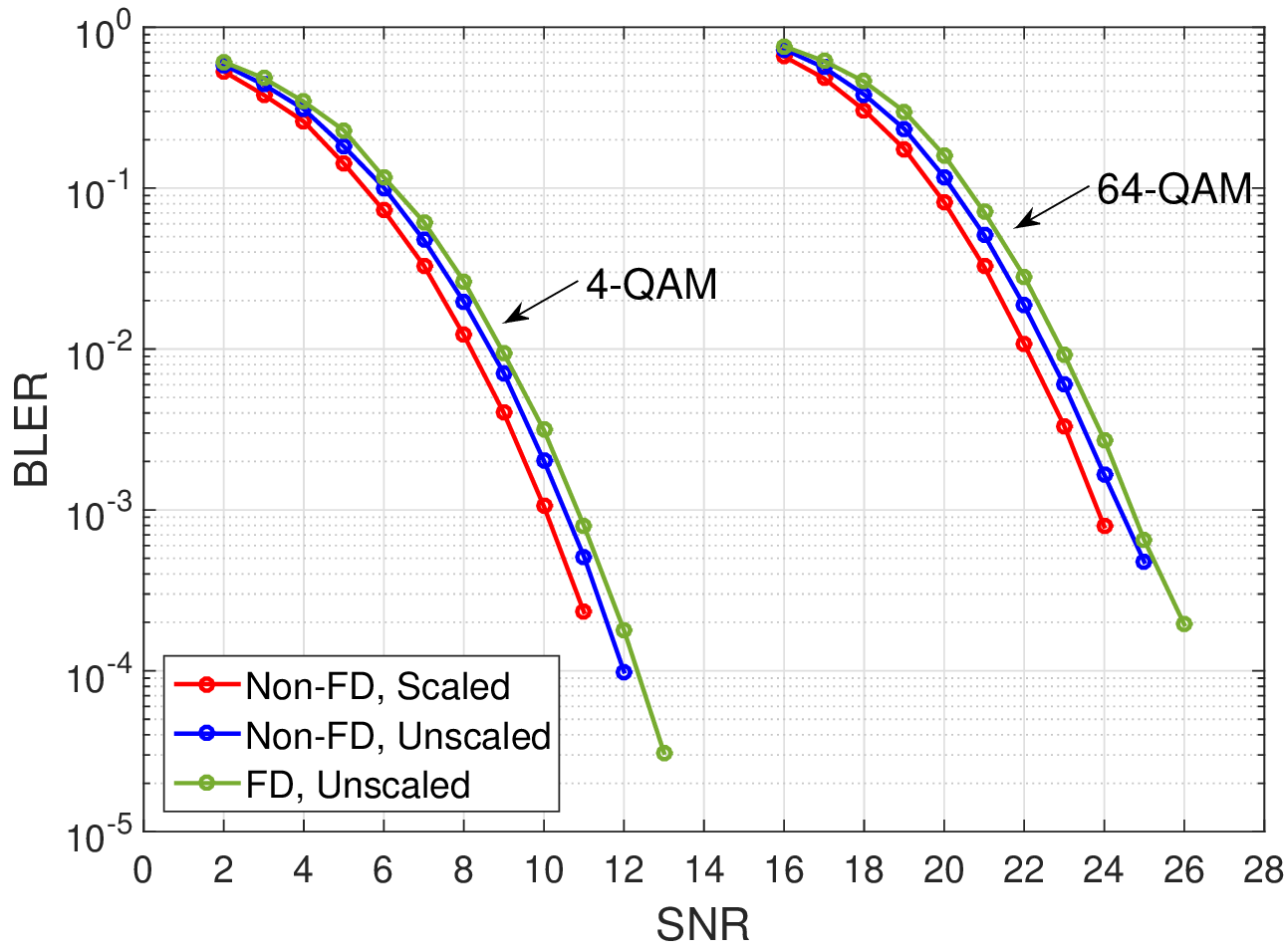}
	\caption{Comparison of an optimal, non-FD code and an example FD code with a complexity reduction of $37.5\%$ for $2$ relays, with 4-QAM (left) and 64-QAM (right) signaling, with $n_s = n_r = 1$, $n_d = 4$.}
	\label{fig:SIMO_FD}
\end{figure}

From Figure~\ref{fig:SIMO_FD}, we see that the proposed FD code and both versions of the lifted Golden code perform comparably for both depicted signaling sets, $4$-QAM and $64$-QAM. The proposed code, however, exhibits a decoding complexity of $|S|^{10}$, while the decoding complexity of the lifted Golden code is $|S|^{16}$, where $S$ is the underlying real signaling alphabet. In exchange for the significantly lower decoding complexity, the FD code shows a loss of merely 1dB, while achieving a similar diversity. We point to Remark~\ref{rmk:improvement} for a short note on possible improvements. 

For the simulations in the MIMO scenario, we remark that no codes can be found in the literature for $N \ge 3$ relays. Thus, we compare the two codes constructed in Example~\ref{exp:mimo_fd} using the method from Theorem~\ref{thm:mult_antenna_relay_code}. For the simulations, we fix $n_d = 6$.  
\begin{figure}[h]
\centering
	\includegraphics[width=0.45\textwidth, height=0.3\textwidth]{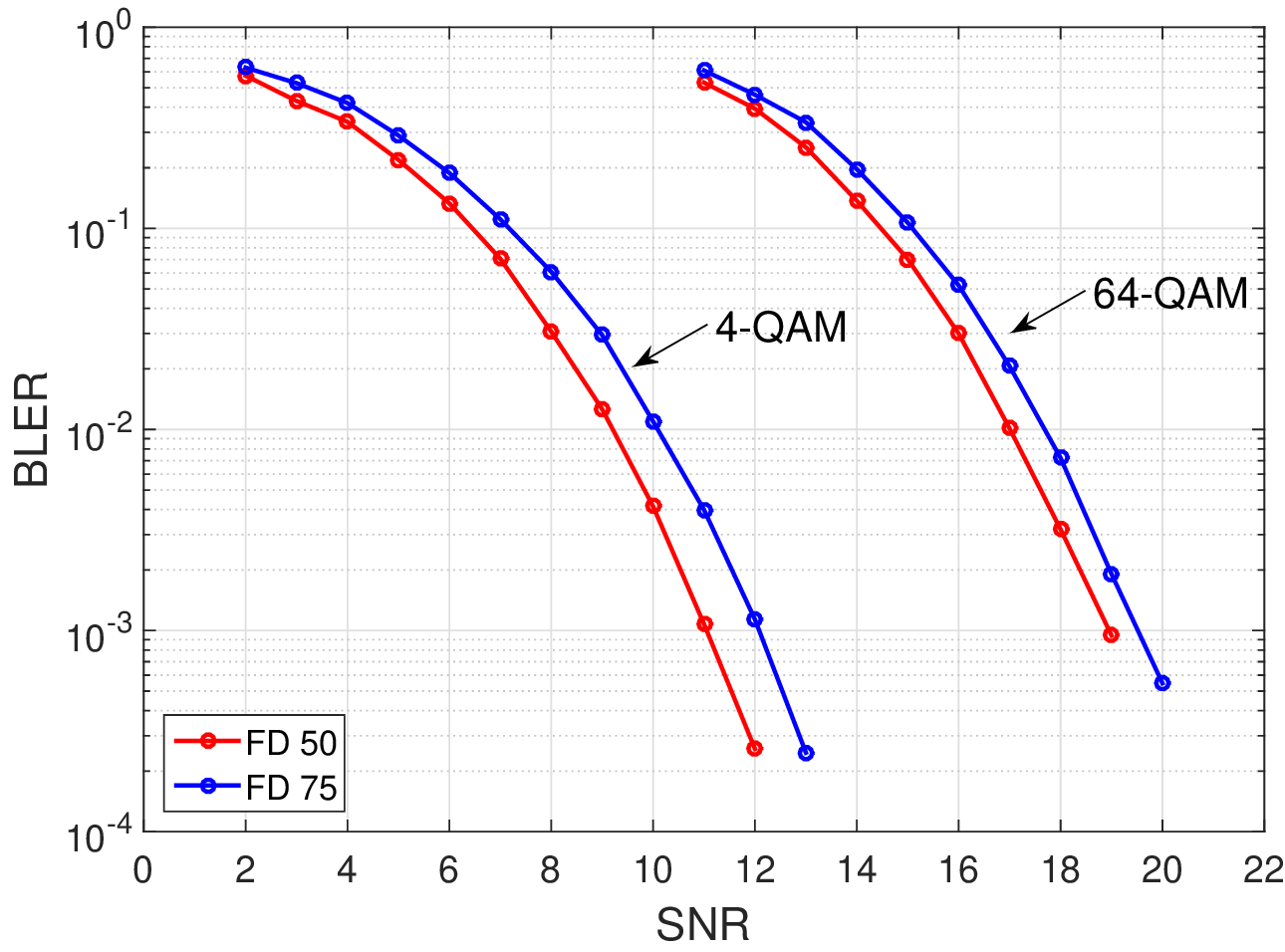}
	\caption{Comparison of two example FD codes with a complexity reduction of $50\%$ (FD 50) and $75\%$ (FD 75), respectively, for $3$ relays, with 4-QAM (left) and 64-QAM (right) signaling, with $n_s + n_r = 4$, $n_d = 6$. }
	\label{fig:MIMO_FD}
\end{figure}
The constructed codes, as can be seen from the figure, perform comparably, exhibiting a similar diversity, and their decoding complexity is $|S|^{12}$ and $|S|^6$, respectively, in contrast to $|S|^{24}$ of a non-FD code of the same rank, where $S$ is the effective real signaling alphabet. The loss in performance of approximately 1 dB is thus traded off against a considerably lower decoding complexity.  

\begin{remark}
	As remarked above, a higher number of receive antennas than needed for the codes to be full-rank was used in order to facilitate the simulations for the larger signaling sets. The relative difference in performance is the same for the minimal required number of antennas, although the relevant $\SNR$, as one would expect, would be shifted to the right.
\end{remark}

\begin{remark}
\label{rmk:improvement}
We want to remark that the codes constructed using the methods in Theorem~\ref{thm:single_antenna_relay_code} and \ref{thm:mult_antenna_relay_code} could be improved by reshaping the underlying lattices, for instance by restricting the elements of the codewords to some suitable ideal, and further by conveniently scaling the codewords. This optimization is however not treated in this article. 
\end{remark}

\section{Fast-Decodable Noncooperative Space--Time Codes}
\label{sec:mac_stc}
In this section we consider the transmissions by $N$ users to a joint destination, \emph{e.g.}, an uplink transmission to a base station. Both the users and the destination can be equipped with multiple antennas. In contrast to the previous section, no cooperation is allowed between the users. 

\subsection{Multiple-Access Channel}
\label{subsec:ma_channel}
 Assuming a noncooperative multi-user communications scenario, the channel is known as either a symmetric or asymmetric \emph{multiple-access channel} (MAC), depending on whether all users are equipped with the same or a different number of transmit antennas. 
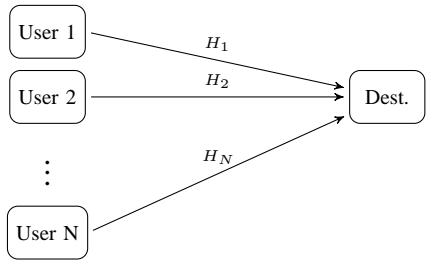
\begin{figure}[h]
\begin{small}
\centering
\begin{tikzpicture}[node distance=1cm]
\node[user] (u1) {\scriptsize User 1};
\node[user, below=0.15cm of u1] (u2) {\scriptsize User 2};
\node[below=0.15cm of u2] (vert) {\scriptsize $\textbf \vdots$};
\node[user, below=0.15cm of vert] (uN) {\scriptsize User N};
\node[user, right=3.5cm of u2] (dest) {\scriptsize Dest.}
	edge[pil_rev] node[pos=0.5, above]{\tiny $H_1$} (u1.east)
	edge[pil_rev] node[pos=0.5, above]{\tiny $H_2$} (u2.east)
	edge[pil_rev] node[pos=0.5, above]{\tiny $H_N$} (uN.east);
\end{tikzpicture}
\caption{System model for the K-user multiple access channel.}
\label{fig:mac_systemfig}
\end{small}
\end{figure}
The transmission of user $K_i$, $i = 1,\ldots,N$, is a point-to-point communication problem over a wireless MIMO channel modeled by a matrix $H_i$ of corresponding size. The ST code matrices of every user are generated independently of those of the remaining users. For the sake of exposure, we assume that every user $K_i$ is equipped with $n_s$ antennas, while the destination is assumed to have  $n_d$ antennas. The disadvantage of having independent code matrices is that the overall ST code does not exhibit a lattice structure, but can still be represented via a set of linear dispersion matrices acting as weight matrices. An important consequence is the following result. 

\begin{proposition}\label{prop:mac_nvd}\cite[Thm.~3]{francis_dmt_mac}
	For any $N > 1$ and $n_s \ge 1$, there do not exist any linear MIMO-MAC codes satisfying the NVD criterion. 
\end{proposition}    

\begin{remark}
	The above proposition motivated the definition of the \emph{Conditional Nonvanishing Determinant} (CNVD) property introduced in \cite{hollanti_gnvd}. A ST code has the CNVD property if its minimum determinant is either zero or bounded from below. This property, accompanied with a suitable code rate, was shown to be sufficient for achieving the optimal  MAC-DMT \cite{francis_dmt_mac}. 
\end{remark}

\subsection{Fast-Decodable MAC Space--Time Codes}
\label{subsec:mac_codes}
Let $K \ge 2$ be the number of transmitters communicating with a single destination. While each of the users might use a different underlying algebraic structure for code construction, the symmetric scenario ensures that the codewords from every user will be of the same dimensions. 

For $k \in \left\{1,\ldots,K\right\}$ consider user $K_k$ employing a ST code $\mathcal{X}_{k}$ carved from a CDA $\mathcal{C}_k = (\mathbf{L}_k/\mathbf{K}_k,\sigma_k,\gamma_k)$ of degree $n$. A codeword $X_k \in \mathcal{X}_k$ is of the form 
\[
X_k = \lambda(x_k) = \sum\limits_{i=1}^{n^2}{s_{k,i}B_{k,i}}
\] 
for some $x_k \in \mathcal{O}_k \subset \mathcal{C}_k$, where $\lambda$ is as in \eqref{eqn:left_reg} or \eqref{eqn:lambda_sim}, $s_{k,i} \in S_k$ are the signaling coefficients and $\left\{B_{k,i}\right\}_{1 \le i \le n^2}$ is the set of weight matrices. Let $\mathbf{F}_k \subset \mathbf{K}_k$ be an intermediate field so that $\mathbf{K}_k/\mathbf{F}_k$ is cyclic Galois of degree $m$ with $\langle \tau_k \rangle = \Gamma(\mathbf{K}_k/\mathbf{F}_k)$. Then, the overall codeword of user $k$ is of the form $U_k = \left[\begin{smallmatrix} X_k & \tau_k(X_k) & \cdots & \tau_k^{m-1}(X_k) \end{smallmatrix}\right]$, and the overall transmitted codeword by all users is 
\[
	X = \left[\begin{smallmatrix} X_1 & \tau_1(X_1) & \cdots & \tau_1^{m-1}(X_1) \\ 
 \vdots & \vdots &  & \vdots \\ X_K & \tau_K(X_K) & \cdots & \tau_K^{m-1}(X_K) \end{smallmatrix}\right].
\]

A set of linear dispersion matrices can be given for the overall code by complementing the corresponding lattice basis of each individual user with zero-matrices of suitable size, namely 
\[
	\mathcal{B} = \left\{\left[\begin{smallmatrix} \mathbf{0}_{(k-1)n} \\ B_{k,i} \\ \mathbf{0}_{(K-k)n}\end{smallmatrix}\right]_{Kn}\right\}_{\substack{1 \le i \le n^2 \\ 1 \le k \le K}} =: \left\{B_{i}\right\}_{1 \le i \le Kn^2}.
\]
The resulting code is not NVD due to Prop.~\ref{prop:mac_nvd}, but is FD and exhibits the CNVD property if the algebraic structures are chosen properly, as a straightforward adaptation of the results about distributed ST codes of the previous section to this noncooperative scenario. 

\begin{example}
	Consider $K = 2$ transmitters equipped with $n_s = 2$ antennas each, and a single destination with $n_d = 4$. Both transmitters carve their ST codes from the quaternion algebra $\mathcal{C}$ from Figure~\ref{fig:siso_tower}, with $a = -3$, $\xi = \imath$, and $m = 2$. The algebra $\mathcal{C}$ is division for $\gamma = -\frac{2}{\sqrt{5}}$. Let $\langle \tau:\imath\mapsto -\imath\rangle = \Gamma(\mathbf{K}/\mathbf{F})$. Then, for $k = 1,2$, codewords are of the form 
$
	U_k = \left[\begin{smallmatrix} X_k & \tau(X_k)\end{smallmatrix}\right],
$
where for $x_k \in \mathcal{O} \subset \mathcal{C}$,
$
X_k = \tilde{\lambda}(x_k) = \left[\begin{smallmatrix} x_{k,1}+x_{k,2}\theta & -\sqrt{-\gamma}(x_{k,3}+x_{k,4}\sigma(\theta)) \\ \sqrt{-\gamma}(x_{k,3}+x_{k,4}\theta) & x_{k,1}+x_{k,2}\sigma(\theta) \end{smallmatrix}\right],
$
with $\theta = \frac{1+\sqrt{-3}}{2}$. The overall transmitted codewords are of the form 
\[
X = \left[\begin{smallmatrix} X_1 & \tau(X_1) \\ X_2 & \tau(X_2) \end{smallmatrix}\right].
\]

For a basis $\left\{b_i\right\}_{1 \le i \le 4} = \left\{1,\imath,\sqrt{-2},\imath\sqrt{-2}\right\}$ of $\mathcal{O}_{\mathbf{K}}$, a set of weight matrices for this code is 
\[
	\left\{B_i\right\}_{1 \le i \le 32} := \left\{\left[\begin{smallmatrix} X_1(b_i,0,0,0) & \tau(X_1(b_i,0,0,0)) \\ X_2(0,0,0,0) & \tau(X_2(0,0,0,0)) \end{smallmatrix}\right],\ldots,\left[\begin{smallmatrix} X_1(0,0,0,0) & \tau(X_1(0,0,0,0)) \\ X_2(0,0,0,b_i) & \tau(X_2(0,0,0,b_i)) \end{smallmatrix}\right]\right\}_{1 \le i \le 4}.
\]

We briefly disclose the performance of the constructed FD MAC code by comparing to two strongly performing codes, which we denote by \emph{NFD1} \cite[Section~IV-D]{francis_mac_constructions}, and \emph{NFD2} \cite[Section~III-E]{francis_mac_constructions}. For fair comparison, since the latter was originally constructed for $n_d = 2$ receive antennas and differs in terms of code rate from the other two codes, we use 16-QAM signaling for NFD2 instead of 4-QAM in order to match the data rate, and allow for $n_d = 4$ antennas at the destination for improved diversity. 
\begin{figure}[h]
\centering
	\includegraphics[width=0.45\textwidth, height=0.3\textwidth]{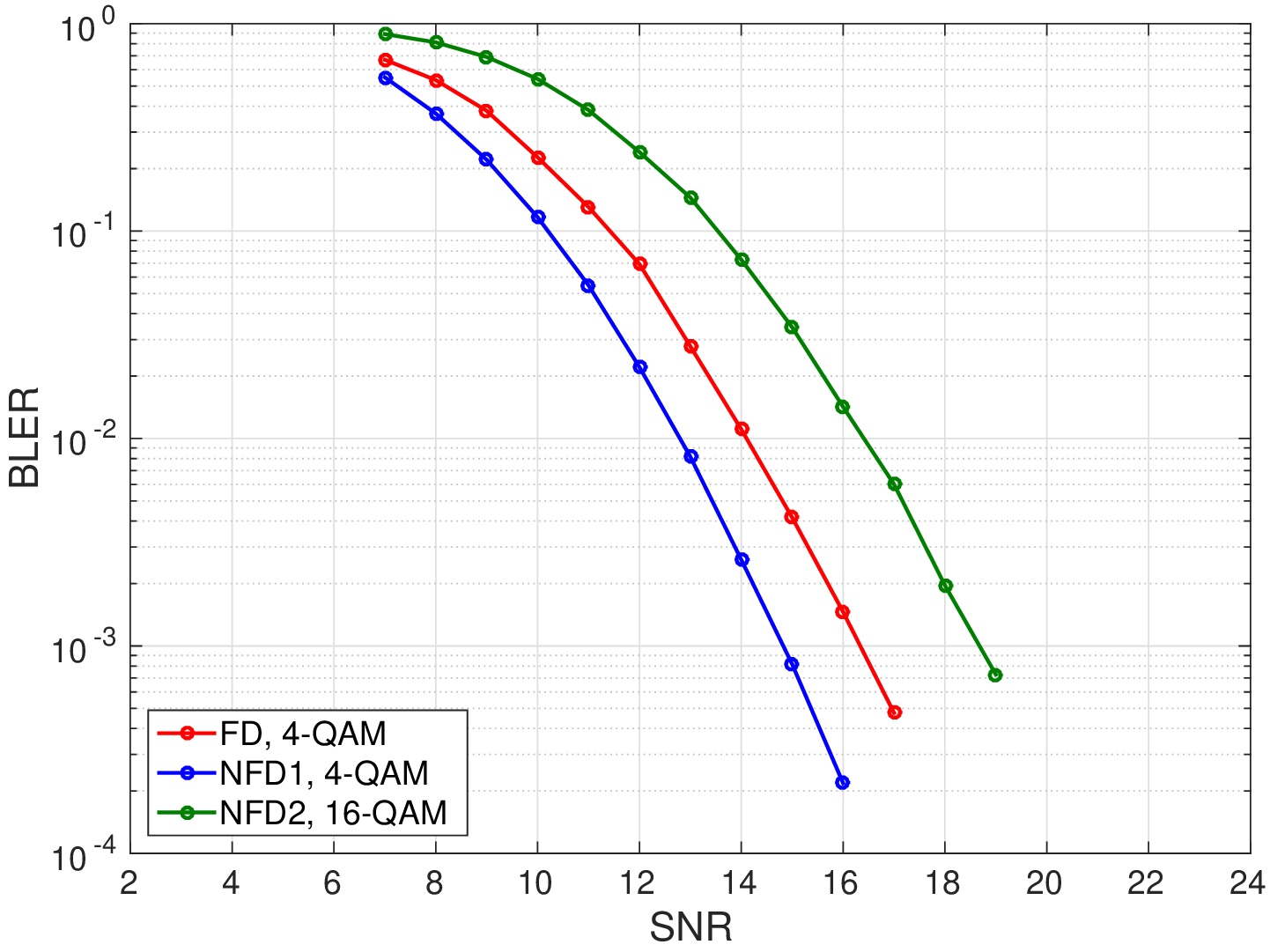}
	\caption{Comparison of two non-FD MAC codes and an example FD MAC code for two transmitters, $n_s = 2$, $n_d = 4$.}
	\label{fig:MAC_FD}
\end{figure}
From Figure~\ref{fig:MAC_FD}, we see that the proposed FD MAC code performs comparably to NFD1, achieving a similar diversity and showing a mere loss of 1-1.5dB, while outperforming NFD2. The proposed code, however, is $2$-GD, exhibiting a decoding complexity of $|S|^{16}$ as opposed to $|S|^{32}$, resulting in a significant reduction of $50\%$ in decoding complexity. Note further that the depicted performance is before any attempts of optimization, and it can be expected that the gap could be narrowed down further without compromising fast-decodability. 
\end{example}

\section{Conclusions}
\label{sec:conclusion}

In this work, we presented two methods for constructing distributed space--time block codes in the setting of amplify-and-forward relaying. We have separately considered a single or multiple antennas at the source and each of the relays, resulting in flexible constructive methods to obtain fast-decodable, more specifically $2$-group and (conditionally) $4$-group decodable, full-diversity space--time codes which have nonvanishing determinants for both cases. The obtained codes can be decoded with a low number of antennas -- in the MIMO case even a single antenna suffices -- and their worst-case decoding complexity is reduced by up to $75\%$, which is known to be the best possible reduction for division algebra based codes. These are highly desirable properties for many applications related to the future 5G networks, such as device-to-device communications and proximity-based services on the wireless edge. 

We have further shown how to use these methods to obtain fast-decodable space--time codes for the $K$-user MIMO-MAC scenario. Although codes for this channel cannot exhibit the nonvanishing determinant property due to the nature of the communications setting, the constructed codes using the methods introduced in this work exhibit the conditional nonvanishing determinant property, which is known to be useful for achieving the optimal MAC-DMT.



\begin{thebibliography}{10}
\providecommand{\url}[1]{#1}
\csname url@samestyle\endcsname
\providecommand{\newblock}{\relax}
\providecommand{\bibinfo}[2]{#2}
\providecommand{\BIBentrySTDinterwordspacing}{\spaceskip=0pt\relax}
\providecommand{\BIBentryALTinterwordstretchfactor}{4}
\providecommand{\BIBentryALTinterwordspacing}{\spaceskip=\fontdimen2\font plus
\BIBentryALTinterwordstretchfactor\fontdimen3\font minus
  \fontdimen4\font\relax}
\providecommand{\BIBforeignlanguage}[2]{{%
\expandafter\ifx\csname l@#1\endcsname\relax
\typeout{** WARNING: IEEEtran.bst: No hyphenation pattern has been}%
\typeout{** loaded for the language `#1'. Using the pattern for}%
\typeout{** the default language instead.}%
\else
\language=\csname l@#1\endcsname
\fi
#2}}
\providecommand{\BIBdecl}{\relax}
\BIBdecl

\bibitem{hollanti_relay1}
C.~Hollanti and N.~Markin, ``A unified framework for constructing
  fast-decodable codes for {N} relays,'' in \emph{Proc. Int. Symp. Math. Theory Netw. Syst.}, 2012.

\bibitem{hollanti_relay2}
------, ``Algebraic fast-decodable relay codes for distributed
  communications,'' in \emph{Proc. IEEE Int. Symp. Inf. Theory}, 2012.

\bibitem{emc}
``Extracting value from chaos,''
  \url{www.emc.com/collateral/analyst-reports/idc-extracting-value-from-chaos-ar.pdf},
  2011, digital Universe Study by the EMC coorporation.

\bibitem{cisco}
``Cisco visual networking index: Global mobile data traffic forecast update,
  2013--2018,'' 2014, white paper.

\bibitem{ericsson}
``{5G} radio access,'' 2013, white paper.

\bibitem{rabaey1}
J.~Rabaey, ``A brand new wireless day,'' in \emph{Keynote Address, ASPDAC,
  Seoul}, 2008.

\bibitem{rabaey2}
{J. Rabaey et al.}, ``Connectivity brokerage -- enabling seamless cooperation
  in wireless networks,'' 2010, white paper.

\bibitem{rabaey3}
J.~Rabaey, ``The swarm at the edge of the cloud - a new perspective on
  wireless,'' in \emph{Proc. Symp. VLSI Circuits}, 2011.

\bibitem{barreal_dss}
A.~Barreal, C.~Hollanti, D.~Karpuk, and H.~f.~Lu, ``Algebraic codes and a new
  physical layer transmission protocol for distributed storage systems,'' in
  \emph{Proc. Int. Symp. Math. Theory Netw. Syst.}, 2014.

\bibitem{hollanti_dss}
C.~Hollanti, H.~f.~Lu, D.~Karpuk, and A.~Barreal, ``New relay-based
  transmission protocols for wireless distributed storage systems,'' in
  \emph{Proc. IEEE Int. Symp. Inf. Theory Appl.}, 2014.

\bibitem{viterbo}
E.~Viterbo and J.~Boutros, ``A universal lattice code decoder for fading
  channel,'' \emph{IEEE Trans. Inf. Theory}, vol.~45, no.~7, pp. 1639--1642,
  1999.

\bibitem{biglieri}
E.~Biglieri, Y.~Hong, and E.~Viterbo, ``On fast-decodable space--time block
  codes,'' \emph{IEEE Trans. Inf. Theory}, vol.~55, no.~2, pp. 524--530, 2009.

\bibitem{jithamithra}
G.~R. Jithamithra and B.~S. Rajan, ``Minimizing the complexity of fast sphere
  decoding of {STBC}s,'' \emph{IEEE Trans. Wirel. Commun.}, vol.~12, no.~12,
  pp. 6142--6153, 2013.

\bibitem{srinath}
K.~Srinath and B.~S. Rajan, ``Low {ML}-decoding complexity, large coding gain,
  full-rate, full-diversity {STBC}s for $2 \times 2$ and $4 \times 2$ {MIMO}
  systems,'' \emph{IEEE J. Sel. Top. Signal Process.}, vol.~3, no.~6, pp.
  916--927, 2009.

\bibitem{markin}
N.~Markin and F.~Oggier, ``Iterated space--time code constructions from cyclic
  algebras,'' \emph{IEEE Trans. Inf. Theory}, vol.~59, no.~9, pp. 5966--5979,
  2013.

\bibitem{vehkalahti}
R.~Vehkalahti, C.~Hollanti, and F.~Oggier, ``Fast-decodable asymmetric
  space--time codes from division algebras,'' \emph{IEEE Trans. Inf. Theory},
  vol.~58, no.~4, pp. 2362--2385, 2012.

\bibitem{berhuy}
G.~Berhuy, N.~Markin, and B.~A. Sethuraman, ``Bounds on fast decodability of
  space time block codes, skew-hermitian matrices, and azumaya algebras,''
  arXiv: 1405.5966.

\bibitem{berhuy2}
------, ``Fast lattice decodability of space--time block codes,'' in \emph{Proc. IEEE Int. Symp. Inf. Theory}, 2014.

\bibitem{ren}
T.~P. Ren, Y.~L. Guan, C.~Yuen, and R.~J. Shen, ``Fast-group-decodable
  space--time block code,'' in \emph{Proc. IEEE Inf. Theory Workshop}, 2010.

\bibitem{kramer}
G.~Kramer and A.~J. van Wijngaarden, ``On the white gaussian multiple-access
  relay channel,'' in \emph{Proc. IEEE Int. Symp. Inf. Theory}, 2000.

\bibitem{chen}
D.~Chen, K.~Azarian, and J.~Laneman, ``A case for amplify-forward relaying in
  the block-fading multiple-access channel,'' \emph{IEEE Trans. Inf. Theory},
  vol.~54, no.~8, pp. 3728--3733, 2008.

\bibitem{nazer}
B.~Nazer and M.~Gastpar, ``Compute-and-forward: Harnessing interference with
  structured codes,'' in \emph{Proc. IEEE Int. Symp. Inf. Theory}, 2008.

\bibitem{laneman}
J.~Laneman and G.~Wornell, ``Distributed space--time-coded protocols for
  exploiting cooperative diversity in wireless networks,'' \emph{IEEE Trans. Inf. Theory}, vol.~49, no.~10, pp. 2415--2425, 2003.

\bibitem{jing}
Y.~Jing and B.~Hassibi, ``Distributed space--time coding in wireless relay
  networks,'' \emph{IEEE Trans. Wirel. Commun.}, vol.~5, no.~12, pp.
  3524--3536, 2006.

\bibitem{kiran}
T.~Kiran and B.~S. Rajan, ``Distributed space--time codes with reduced decoding
  complexity,'' in \emph{Proc. IEEE Int. Symp. Inf. Theory}, 2006.
  
\bibitem{rajan}
G.~S. Rajan and B.~S. Rajan, ``Multigroup {ML} decodable collocated and
  distributed space--time block codes,'' \emph{IEEE Trans. Inf. Theory},
  vol.~56, no.~7, pp. 3221 -- 3247, 2010.

\bibitem{barreal_icmcta}
A.~Barreal, C.~Hollanti, and N.~Markin, ``Constructions of Fast-Decodable 
  Distributed Space--Time Codes,'' \emph{Proc. 4th Int. Castle Meet. Coding Theory Appl.}, CIM Series in Mathematical 
  Sciences, Springer-Verlag, 2014.
  
\bibitem{hollanti_mindet}
C.~Hollanti and H.~f.~Lu, ``Normalized minimum determinant calculation for
  multi-block and asymmetric space--time codes,'' in \emph{Applied Algebra,
  Algebraic Algorithms and Error-Correcting Codes}, ser. Lecture Notes in
  Computer Science.\hskip 1em plus 0.5em minus 0.4em\relax Springer Berlin
  Heidelberg, 2007, vol. 4851, pp. 227--236.

\bibitem{hollanti_order1}
C.~Hollanti and J.~Lahtonen, ``A new tool: Constructing {STBC}s from maximal
  orders in central simple algebras,'' in \emph{Proc. IEEE Int. Theory Workshop}, 2006.

\bibitem{hollanti_order2}
C.~Hollanti, J.~Lahtonen, and H.~f.~Lu, ``Maximal orders in the design of dense
  space--time lattice codes,'' \emph{IEEE Trans. Inf. Theory}, vol.~54, no.~10,
  pp. 4493--4510, 2008.

\bibitem{vehkalahti2}
R.~Vehkalahti, C.~Hollanti, J.~Lahtonen, and K.~Ranto, ``On the densest {MIMO}
  lattices from cyclic division algebras,'' \emph{IEEE Trans. Inf. Theory},
  vol.~55, no.~8, pp. 3751--3780, 2009.

\bibitem{oggier_perfect}
F.~Oggier, G.~Rekaya, J.-C. Belfiore, and E.~Viterbo, ``Perfect space time
  block codes,'' \emph{IEEE Trans. Inf. Theory}, vol.~52, no.~9, pp.
  3885--3902, 2006.

\bibitem{sethuraman}
B.~A. Sethuraman, B.~S. Rajan, and V.~Shashidhar, ``Full-diversity, high-rate
  space--time block codes from division algebras,'' \emph{IEEE Trans. Inf.
  Theory}, vol.~49, no.~10, pp. 2596--2616, 2003.

\bibitem{belfiore}
J.-C. Belfiore, G.~Rekaya, and E.~Viterbo, ``The golden code: a $2\times 2$
  full-rate space--time code with non-vanishing determinants,'' \emph{IEEE
  Trans. Inf. Theory}, vol.~51, no.~4, pp. 1432--1436, 2005.

\bibitem{hollanti_thesis}
C.~Hollanti, ``Order-theoretic methods for space--time coding: symmetric and
  asymmetric designs,'' Ph.D. dissertation, University of Turku, 2009.

\bibitem{nabar}
R.~U. Nabarand, H.~B{\"o}lcskei, and F.~W. Kneub{\"u}hler, ``Fading relay
  channels: performance limits and space--time signal design,'' \emph{IEEE J.
  Sel. Areas Commun.}, vol.~22, no.~6, pp. 1099--1109, 2004.

\bibitem{yang}
S.~Yang and J.-C. Belfiore, ``Optimal space--time codes for the {MIMO}
  amplify-and-forward cooperative channel,'' \emph{IEEE Trans. Inf. Theory},
  vol.~53, no.~2, pp. 647--663, 2007.

\bibitem{marcus}
D.~A. Marcus, \emph{Number Fields}.\hskip 1em plus 0.5em minus 0.4em\relax New
  York: Springer-Verlag, 1977.

\bibitem{unger}
T.~Unger and N.~Markin, ``Quadratic forms and space--time block codes from
  generalized quaternion and biquaternion algebras,'' \emph{IEEE Trans. Inf.
  Theory}, vol.~57, no.~9, pp. 6148--6156, 2011.

\bibitem{francis_dmt_mac}
H.~f.~Lu, C.~Hollanti, R.~Vehkalahti, and J.~Lahtonen, ``{DMT} optimal codes
  constructions for multiple-access {MIMO} channel,'' \emph{IEEE Trans. Inf.
  Theory}, vol.~57, no.~6, pp. 3594--3617, 2011.

\bibitem{hollanti_gnvd}
C.~Hollanti, H.~f.~Lu, and R.~Vehkalahti, ``An algebraic tool for obtaining
  conditional non-vanishing determinants,'' in \emph{Proc. IEEE Int. Symp. Inf. Theory}, 2009.

\bibitem{zheng}
L.~Zheng and D.~Tse, ``Diversity and multiplexing: a fundamental tradeoff in
  multiple-antenna channels,'' \emph{IEEE Trans. Inf. Theory}, vol.~49, no.~5,
  pp. 1073--1096, 2003.
  
\bibitem{francis_mac_constructions}
H.~f.~Lu, R.~Vehkalahti, C.~Hollanti, J.~Lahtonen, Y.~Hong, and E.~Viterbo, ''New space--time 
code constructions for two-user multiple access channels,'' \emph{IEEE J. Sel. Top. Signal Process.},
vol.~3, no.~6, pp. 939--957, 2009.

\end{thebibliography}
\end{document}